\documentclass[a4paper,USenglish,autoref,thm-restate]{lipics-v2021}

\usepackage{algorithm}
\usepackage[noend]{algpseudocode}
\usepackage{amsmath,amssymb,amsthm}
\usepackage{alltt}
\usepackage{array}
\usepackage{booktabs}
\usepackage[commandnameprefix=ifneeded]{changes}  
\usepackage{cite}
\usepackage{enumitem}
\usepackage{fontawesome5}
\usepackage{graphicx}
\usepackage{multirow}
\usepackage[autolanguage]{numprint} \npthousandsep{,}
\usepackage{subcaption}
\usepackage{textcomp}
\usepackage{tikz} \usetikzlibrary{arrows.meta,decorations.pathmorphing,graphs.standard,positioning}
\usepackage{xcolor}

\graphicspath{{figs/}{./}}

\renewcommand{\Pr}[1]{\ensuremath{\mathrm{Pr}\left[#1\right]}}
\newcommand{\E}[1]{\ensuremath{\mathrm{E}\left[#1\right]}}
\newcommand{\bigO}[1]{\ensuremath{\mathcal{O}(#1)}}

\newcommand{\true}{\textsc{true}}
\newcommand{\false}{\textsc{false}}
\newcommand{\alg}{\ensuremath{\mathcal{A}}}

\makeatletter
\algrenewcommand\ALG@beginalgorithmic{\small}
\algrenewcommand\alglinenumber[1]{\footnotesize #1:}
\makeatother

\newif\ifcomment
\commenttrue    

\newif\iffigabbrv
\figabbrvfalse   
\newcommand{\figtext}{\iffigabbrv Fig.\else Figure\fi}

\newif\ifeqabbrv
\eqabbrvtrue    
\newcommand{\eqtext}{\ifeqabbrv Eq.\else Equation\fi}
\newcommand{\eqstext}{\ifeqabbrv Eqs.\else Equations\fi}

\title{Adaptive Self-Organization in Anonymous Dynamic Networks}
\titlerunning{Adaptive Self-Organization in Anonymous Dynamic Networks}

\author{Garrett Parzych}{School of Computing and Augmented Intelligence \and Biodesign Center for Biocomputing, Security and Society\\Arizona State University, Tempe, AZ, USA}{gparzych@asu.edu}{https://orcid.org/0009-0008-4789-9603}{}
\author{Joshua J. Daymude}{School of Computing and Augmented Intelligence \and Biodesign Center for Biocomputing, Security and Society\\Arizona State University, Tempe, AZ, USA}{jdaymude@asu.edu}{https://orcid.org/0000-0001-7294-5626}{}
\authorrunning{G.\ Parzych and J.\ J.\ Daymude}

\Copyright{Garrett Parzych and Joshua J.\ Daymude}



\funding{This work is supported in part by National Science Foundation award CCF-2312537.}

\acknowledgements{J.J.D.\ thanks Annie Daymude, Monty Carson, and Julie Carson for the many hours of childcare that created the working time to complete this paper.}

\hideLIPIcs
\nolinenumbers

\begin{document}


\maketitle

\begin{abstract}
    We introduce the problem of \textit{adaptive self-organization} in which the nodes of an anonymous, synchronous dynamic network must distributively change the collective distribution of their responses (or ``colors'') as a function of \textit{time-varying environmental signals}, even when these signals are only perceived locally and the network topology changes adversarially.
    Specifically, a \textit{signal adversary} may change the type of signal and which node(s) \textit{witness} that signal arbitrarily between rounds.
    If a signal (or lack thereof) $s$ persists in the system for sufficiently long, the dynamic network must \textit{stabilize} such that nodes' colors closely approximate $r(s)$, a goal distribution defined by the problem instance.
    By symmetry, deterministic nodes can only hope to solve \textit{homogeneous} instances of adaptive self-organization, i.e., those in which all nodes stabilize with the same color.
    We present a linear-time, logarithmic-memory, deterministic algorithm for this class of instances that works even when the multiplicity and location of signal witnesses change arbitrarily.
    We then give a randomized extension of this algorithm that solves arbitrary (i.e., not necessarily homogeneous) instances of adaptive self-organization with high probability in the same time and space bounds.
\end{abstract}

\section{Introduction} \label{sec:intro}

The theory of dynamic networks analyzes nodes' abilities to distributively solve problems of coordination and communication even as the edges between them change over time~\cite{Altisen2023-selfstabilizingsystems,Casteigts2012-timevaryinggraphs}.
Many fundamental problems such as broadcast, leader election, and consensus have been studied in these models~\cite{Augustine2016-distributedalgorithmic,Casteigts2018-journeydynamic}, but always as isolated goal behaviors.
In this paper, we consider an orthogonal dimension of dynamics, not only of changes to a network's topology, but also of time-varying stimuli that a network must respond to.
We are specifically interested in collectives' fundamental ability to achieve \textit{adaptive self-organization} of their members among a set of discrete behaviors in response to \textit{time-varying environmental signals}, even when these signals may only be perceived locally by a few members and the possible interactions among members are dynamic.

Biological distributed systems rarely have just one collective function.
Instead, members respond to locally-perceived environmental signals by altering their individual behaviors which may then induce a collective adaptation to the environment at the system level.
For example, the discovery of a food source in an ant colony's vicinity triggers a shift in foraging behaviors from exploration to gathering; its depletion does the opposite~\cite{Carroll1973-ecologyforaging,Flanagan2012-quantifyingeffect}.
When slime molds (i.e., \textit{Dictyostelium discoideum}) experience nutrient scarcity, most of their cells aggregate to form a fruiting body of spores while other ``loners'' stay behind to exploit the possible return of food to the area~\cite{Tarnita2015-fitnesstradeoffs,Rossine2020-ecoevolutionarysignificance}.
In mammalian immune systems, helper T cells secrete different combinations of cytokines based on the type of invading pathogen (i.e., parasites, fungi, viruses, or bacteria) to coordinate the appropriate immune response and then return to a passive state of immune memory after the pathogen is cleared~\cite{Sompayrac2019-howimmune}.
Each of these biological collectives fluidly adapts among response behaviors in the presence of stimuli and gracefully ``resets'' to a baseline behavior in their absence.

With this inspiration, we propose and analyze the problem of \textit{adaptive self-organization} in anonymous dynamic networks.
In each round, the network's environment potentially emits a \textit{signal} that is then \textit{witnessed} by some subset of nodes.
If a particular signal persists in the system for sufficiently long, the nodes must self-organize their \textit{colors}---representing different behaviors or responses---to approximate that signal's goal distribution as defined by the problem instance.
Likewise, if there is no signal for sufficiently long, nodes must ``reset'' to a baseline distribution of colors.
All of this must be achieved despite adversarial dynamics in both the network topology and signals.

\subparagraph{Our Contributions.}

We consider distributed algorithms run by nodes that are anonymous (lacking identifiers) and port-indistinguishable (lacking port labels) in a dynamic network whose topology can change arbitrarily between synchronous rounds but always remains connected (1-interval connected).
In this setting, we prove:
\begin{itemize}[itemsep=0pt]
    \item Any instance of adaptive self-organization requiring a network to reach and remain in a configuration with multiple colors (i.e., approximating a \textit{non-homogeneous} distribution) is impossible for deterministic algorithms, regardless of dynamics (Theorem~\ref{thm:hardness}).

    \item There exists a linear-time, logarithmic-memory, deterministic algorithm that solves any \textit{homogeneous} instance of adaptive self-organization, i.e., one in which every signal maps to a goal distribution where all nodes have the same color (Theorem~\ref{thm:deterministic}). 
    This algorithm stabilizes in linear time when resetting in the absence of any signal, and stabilizes with at most a linear number of disruptions (i.e., dynamics may temporarily disrupt otherwise stable configurations in a linear number of rounds) when responding to a signal's presence.

    \item If nodes are additionally granted knowledge of $n$, the number of nodes in the network, then a small modification of the deterministic algorithm achieves stabilization in linear time for both responding to a signal's presence and resetting in its absence (Corollary~\ref{cor:knowledgen}).

    \item There exists a randomized algorithm that solves any instance of adaptive self-organization---homogeneous or otherwise---with high probability (Corollary~\ref{cor:randomized}).
    This algorithm extends the deterministic approach and thus inherits its linear time complexity, logarithmic memory bound, and stabilization guarantees.
\end{itemize}

\subsection{Model} \label{subsec:model}

\subparagraph{Dynamic Networks.}

We consider synchronous dynamic networks comprising a fixed set of nodes $V$.
Nodes communicate with each other via message passing over a communication graph whose topology changes over time.
We model this topology as a \textit{time-varying graph} $\mathcal{G} = (V, E, T, \rho)$ where $V$ is the set of nodes, $E$ is the static set of undirected edges that may appear in the graph, $T = \mathbb{N}$ is the lifetime of the graph, and $\rho : E \times T \to \{0,1\}$ is the presence function indicating whether an edge exists at a given time~\cite{Casteigts2012-timevaryinggraphs}.
We refer to the set of edges present at time $t \in T$ as $E_t = \{e \in E: \rho(e, t) = 1\}$ and the undirected graph $G_t = (V, E_t)$ as the \textit{snapshot} of $\mathcal{G}$ at time $t \in T$.
We assume an adversary controls the presence function $\rho$ and that $E$ is the complete set of edges on $V$; i.e., we do not limit which edges the adversary can introduce.
We do, however, follow the standard assumption of \textit{1-interval connectivity}; i.e., the adversary may make arbitrary topological changes at each time $t \in T$ so long as each snapshot $G_t$ is connected.

\subparagraph{Node Capabilities.}

We assume nodes are \textit{anonymous}, lacking unique identifiers, and---unless explicitly stated otherwise---have \textit{no knowledge or approximation of any global measure}, including the number of nodes $n = |V|$.
Said another way, nodes are only capable of executing \textit{uniform} algorithms that are the same regardless of any global parameters, like $n$.
We further assume that nodes have \textit{no port labels}; i.e., they cannot count or locally distinguish among their neighbors.
Consequently, when a node communicates with its neighbors via \textit{message passing}, it must send the same message to all its current neighbors.

\subparagraph{Algorithms and Execution.}

Each node in the time-varying graph $\mathcal{G}$ synchronously executes the same distributed algorithm $\alg$.
All nodes are initialized at time $t = 0$, and each synchronous round $t$ starting at time $t$ proceeds as follows (see also \figtext~\ref{fig:roundstructure}):
\begin{enumerate}
    \item The adversary fixes the network topology $G_t$ and the presence, type, and/or witness(es) of the signal (see Section~\ref{subsec:problem}) for round $t$.
    We allow the adversary to know the state of all nodes at time $t$ when it makes its choices.
    Note that this lets the adversary adapt to decisions made by randomized algorithms.

    \item Each node may send a message to its neighbors in $G_t$ according to $\alg$ as a function of its current state only; messages do not explicitly depend on the presence or type of signal.

    \item Each node may perform a state transition according to $\alg$ as a function of its current state, the multiset of messages it (reliably) receives from its neighbors in $G_t$, and the presence and type of signal it currently witnesses, if any (again, see Section~\ref{subsec:problem}).
\end{enumerate}

\begin{figure}[t]
    \centering
    \resizebox{\textwidth}{!}{\begin{tikzpicture}
        \def\r{5}

        \foreach \t/\x in {0/0,1/1,2/2,t/2.5,t+1/3.5}{
            \draw[dashed] (\r*\x, 0) -- (\r*\x, 4.25);
            \node at (\r*\x, 4.5) {time $\t$};
        }

        \tikzset{round/.style={arrows={Bracket[scale length=3.5,scale width=2]-Arc Barb[scale=2]}}}
        \foreach \t/\x in {0/0,1/1,t/2.5}
            \draw[round] (\r*\x, 0) to node[below] {round $\t$} (\r*\x+\r, 0);

        \tikzset{box/.style={rectangle,draw,align=center,inner xsep=0pt,text width=2.5cm,minimum height=1.2cm}}
        \foreach \x in {1.25,6.25,13.75} {
            \node[box] at (\x, 1) {send/recv.\ messages};
            \node[box] at (\x+2.5, 1) {perform state transition};
        }

        \node at (2.25*\r, 2.25) {$\cdots$};
        \node at (3.75*\r, 2.25) {$\cdots$};

        \node at (2.5, 4.5) {$G_0$};
        \begin{scope}[xshift=2.5cm,yshift=3cm]
            \graph {
                subgraph I_n[nodes={draw,circle},clockwise,radius=1cm,empty nodes,n=7],
                {[path] 1, 2, 4, 7, 6},
                {[path] 3, 4, 5, 2, 7},
                "$s_0$"[at={(1.5, 1)}] ->[thick,blue!70!black,decorate,decoration={snake,post length=0.07cm}] {1, 3};
            };
        \end{scope}
        \node at (7.5, 4.5) {$G_1$};
        \begin{scope}[xshift=7.5cm,yshift=3cm]
            \graph {
                subgraph I_n[nodes={draw,circle},clockwise,radius=1cm,empty nodes,n=7],
                {[path] 3, 1, 5, 7, 4},
                {[path] 7, 3, 4, 6, 1, 2},
                "$s_1 = \bot$"[at={(-1.7, -0.9)}];
            };
        \end{scope}
        \node at (15, 4.5) {$G_t$};
        \begin{scope}[xshift=15cm,yshift=3cm]
            \graph {
                subgraph I_n[nodes={draw,circle},clockwise,radius=1cm,empty nodes,n=7],
                {[path] 3, 5, 4, 2, 6, 7, 1},
                "$s_t$"[at={(-1.8, 0.8)}] ->[thick,red!70!black,decorate,decoration={snake,post length=0.07cm}] {6};
            };
        \end{scope}
    \end{tikzpicture}}
    \caption{Synchronous round structure and algorithm execution as defined in Section~\ref{subsec:model}.}
    \label{fig:roundstructure}
\end{figure}

\subsection{Adaptive Self-Organization} \label{subsec:problem}

\subparagraph{Signals and Witnesses.}

In each round $t$, the environment---represented by a \textit{signal adversary}---emits either one of $k$ distinct signals $s_t \in \{1, \ldots, k\}$ or no signal at all, denoted by $s_t = \bot$.
We assume that the number of signals is a fixed constant $k$, independent of $n$.
If a signal $s_t \neq \bot$ is present in the system, the signal adversary chooses a non-empty set of \textit{witnesses} $W_t \subseteq V$ that perceive it.
Thus, whether a signal exists in the system, which signal is emitted, and what nodes witness an emitted signal may all change arbitrarily and independent of each other between rounds.
A signal $s_t \in \{\bot, 1, \ldots, k\}$ is \textit{persistent} if $s_{t'} = s_t$ for all rounds $t' \geq t$; note that the witnesses $W_t$ may still change.

\subparagraph{Responses (``Colors'').}

Each node may change its \textit{color} over time to indicate its response to different signals (or lack thereof).
Formally, we denote by $c_t(v) \in \{1, \ldots, \ell\}$ the color of node $v \in V$ at time $t$ (i.e., at the start of round $t$) and by $C_t(i)$ the number of nodes at time $t$ with color $i$.
Color can be thought of as some component or function of a node's state, similar to the way ``outputs'' are modeled in Di Luna and Viglietta's work on history trees~\cite{DiLuna2022-computinganonymous}.
We assume the number of colors $\ell$ is a fixed constant, independent of $n$.

\subparagraph{Problem Statement.}

Let $F_\ell = \{(f_1, \ldots, f_\ell) \in \mathbb{Q}^\ell \mid \sum_{i=1}^\ell f_i = 1 \text{ and } f_i \geq 0, \forall i\}$ be all rational probability distributions over the $\ell$ colors.\footnote{We consider rational probability distributions in $\mathbb{Q}^\ell$ rather than real probability distributions in $\mathbb{R}^\ell$ so that algorithms can reasonably encode the proportions $f_i$ in a constant number of bits w.r.t.\ $n$.}
An instance of the \textit{adaptive self-organization problem} takes the form of a response function $r : \{\bot, 1, \ldots, k\} \to F_\ell$ mapping each distinct signal and the absence of a signal to a goal distribution over colors.
Informally, if a signal (or lack thereof) $s$ is persistent, the nodes must collectively change their colors to reach and remain in a distribution approximating $r(s)$.
Formally, suppose a signal $s$ becomes persistent at time $t_0$.
An algorithm execution \textit{stabilizes (w.r.t.\ signal $s$)} if there exists a time $t_1 \geq t_0$ such that $c_{t}(v) = c_{t-1}(v)$ for all times $t > t_1$ and nodes $v \in V$.
A dynamic network \textit{approximates} a goal distribution $r(s)$ at time $t$ if its nodes' color frequencies $C_t/n$ satisfy
\begin{equation} \label{eq:approximate}
    d_{TV}(C_t/n, r(s)) = \frac{1}{2}\sum_{i=1}^\ell \left|\frac{C_t(i)}{n} - r(s, i)\right| < \varepsilon(n),
\end{equation}
where $\varepsilon(n)$ is an error term satisfying $\lim_{n \to \infty} \varepsilon(n) = 0$.
An algorithm \textit{solves} an instance $r$ of adaptive self-organization if, for all sufficiently large $n$, every execution stabilizes w.r.t.\ any persistent signal $s$ in a configuration approximating $r(s)$.
For randomized algorithms, we allow \eqtext~\ref{eq:approximate} to hold \textit{with high probability (w.h.p.)}, i.e., with probability at least $1 - 1/\text{poly}(n)$.

\subparagraph{Time and Space Complexity.} 

As usual for distributed systems, we are interested in the asymptotic growth of complexity measures as a function of $n = |V|$, the number of nodes.
Because of a subtle interaction between the dynamics and signal adversaries, we define two measures of time complexity.
The first is the more natural of the two: we say an algorithm $\alg$ \textit{stabilizes for a signal $s \in \{\bot, 1, \ldots, k\}$ in $T(n)$ rounds} if, given $s$ is persistent at time $t$, every execution of $\alg$ stabilizes w.r.t.\ $s$ by round $t + T(n)$.
With this complexity measure, the algorithm must stabilize within $T(n)$ rounds regardless of how the adversary changes the network topology or signal witnesses over time.
We also introduce a weaker form of stabilization, where otherwise ``stable'' configurations may be temporarily disrupted---but only for a bounded number of rounds before permanent stabilization is reached.
Call a dynamic network \textit{stable in round $t$} if $c_t(v) = c_{t-1}(v)$ for all nodes $v \in V$; note that this one-round notion is strictly weaker than stabilization which requires colors to remain fixed for all future times.
An algorithm $\alg$ \textit{stabilizes for a signal $s \in \{\bot, 1, \ldots, k\}$ with $T(n)$ disruptions} if, given $s$ is persistent at time $t$,
there are at most $T(n)$ not necessarily consecutive rounds after round $t$ in which the dynamic network is not stable.
We note that stabilization in linear time is worst-case optimal in this model.
In particular, the adversary could choose a static path network on $n$ nodes, and if a persistent signal is witnessed at one endpoint, information about that signal cannot reach the other endpoint in fewer than $n-1$ rounds.

An algorithm's \textit{space complexity} is the maximum number of bits a node uses to store its state between rounds.
We emphasize that even if nodes have $\Omega(\log n)$ memory---sufficient for storing unique identifiers---they are anonymous and are not assigned such identifiers a priori.
We do not analyze the complexity of message sizes directly, as our execution model assumes nodes send messages based only on their states; thus, there are at most as many message types as states.

\subsection{Related Work} \label{subsec:relwork}

Adaptive self-organization is fundamentally a communication problem in that nodes disseminate local knowledge (i.e., the presence and type of signal they witness) to coordinate the appropriate response across an entire dynamic network.
But related works on broadcast~\cite{ODell2005-informationdissemination,Parzych2024-memorylower,Hussak2023-terminationamnesiac}, $k$-token dissemination~\cite{Dinitz2018-smoothedanalysis,Dinitz2022-smoothedanalysis}, and aggregate computation~\cite{Michail2013-namingcounting,DiLuna2014-countinganonymous,DiLuna2014-consciousunconscious,DiLuna2016-nontrivial,Chakraborty2018-fasterexactcounting,Kowalski2020-polynomialcounting,DiLuna2022-computinganonymous,Kowalski2022-efficientdistributed,DiLuna2023-optimalcomputation} in anonymous dynamic networks are not directly applicable here.
Unlike most algorithms which use progress as a structural feature---e.g., tracking informed vs.\ uninformed nodes in broadcast or local symmetries that remain to be broken in aggregate computation---our problem allows signals to change or disappear at any moment, suddenly requiring past progress be undone in pursuit of a different goal.
In this sense, our algorithms feature a kind of self-stabilization~\cite{Zhang2024-faulttolerantconsensus,DiLuna2025-universalfinitestate}, though our concern is achieving continuously changing goals from in-progress configurations rather than recovering from arbitrary corruption of node memory. 

Most relevant to our work are Oh, Randall, and Richa's recent results on ``adaptive collective responses'' which are also driven by local stimuli and operate on anonymous dynamic networks~\cite{Oh2025-adaptivecollective}.
If a stimulus (resp., no stimulus) is witnessed for sufficiently long, they desire for all nodes to become \textsc{Aware} (resp., \textsc{Unaware}), akin to a one-signal, two-color, homogeneous version of our problem.
In their randomized solution, witnesses continuously generate ``alert'' tokens that spread through the network using degree-weighted random walks and cause nodes they reach to become \textsc{Aware}.
Whenever a node stops witnessing stimuli, it broadcasts an ``all clear'' token that causes other nodes to become \textsc{Unaware}.
As will also be the case with our algorithm, the main challenge is ensuring the most recent information ``wins out'' over older information, even though both types of tokens can be present in the system simultaneously.
Their algorithm leverages the relative speeds of random walks vs.\ broadcasts to solve the problem in constant memory, $\bigO{n^2}$ expected time for becoming \textsc{Unaware}, and $\bigO{n^6\log n}$ expected time for becoming \textsc{Aware}.

However, Oh et al.\ make several assumptions that are incompatible with ours: (1) their scheduler sequentially activates nodes uniformly at random while we assume synchrony, (2) they restrict dynamics such that each \textsc{Aware} node's set of \textsc{Aware} neighbors always form a connected component while we assume more general 1-interval connectivity, (3) they assume the number of active stimulus locations and the dynamic degree are both upper bounded by constants known to the nodes a priori while we make no such assumptions, and (4) they guarantee all nodes will become \textsc{Aware} if the stimuli remain unchanged for sufficiently long while we allow the number and locations of witnesses to change arbitrarily so long as at least one signal is present anywhere in the network for sufficiently long.
In comparison, our algorithms solve a more general formulation of the problem under more standard assumptions in linear (worst-case optimal) time, but require logarithmic memory.

\section{Hardness of Non-Homogeneous Instances} \label{sec:hardness}

In this section, we prove that instances of adaptive self-organization requiring a network to approximate distributions with multiple colors are impossible for deterministic algorithms.
Formally, an instance $r : \{\bot, 1, \ldots, k\} \to F_\ell$ of adaptive self-organization is \textit{homogeneous} if each signal $s \in \{\bot, 1, \ldots, k\}$ has a color $c_s \in \{1, \ldots, \ell\}$ such that
\begin{equation}
    r(s, i) = \left\{ \begin{array}{cl}
        1 & \text{if $i = c_s$}; \\
        0 & \text{otherwise.}
    \end{array} \right.
\end{equation}
An instance $r$ is \textit{non-homogeneous} otherwise, i.e., if there exists a signal $s$ whose distribution $r(s)$ places non-zero probability mass on multiple colors or, equivalently, if there exists a signal $s$ and a color $i$ such that $0 < r(s, i) < 1$.

The intuition behind the impossibility result below is straightforward: if all nodes (1) are arranged in a highly symmetric graph, (2) witness the same signal, and (3) receive the same messages from their neighbors, then---because they are deterministic---they must always behave in the same ways.
In particular, they must always have the same color.

\begin{restatable}{theorem}{thmhardness} \label{thm:hardness}
    No non-homogeneous instance of adaptive self-organization can be solved by a deterministic algorithm, even in static networks.
\end{restatable}
\begin{proof}
    Consider any non-homogeneous instance $r$ of adaptive self-organization; w.l.o.g., let signal $s \in \{\bot, 1, \ldots, k\}$ and color $c \in \{1, \ldots, \ell\}$ satisfy $0 < r(s, c) < 1$. 
    Suppose to the contrary that there exists a deterministic algorithm $\alg$ that solves this instance, and consider the execution of $\alg$ on a dynamic network of $n$ nodes in which $G_t = K_n$, $s_t = s$, and $W_t = V$ for all rounds $t \geq 0$.
    In words, the network has a static, complete topology; signal $s$ is persistent from the very first round; and $s$ is always witnessed by all nodes.

    Argue by induction on $t \geq 0$ that at time $t$, all nodes have the same state.
    All nodes are initialized in the same state.
    So suppose all nodes have the same state at time $t \geq 0$.
    Since $\alg$ is deterministic, $G_t$ is a complete graph, and all nodes witness the same signal $s$, symmetry demands all nodes send each other identical messages and then use the $n - 1$ identical messages they receive to transition to the same state, as required.

    Since color is a component of state, this implies all nodes share the same color $c_t$ at all times $t$.
    Moreover, since $\alg$ solves this instance and signal $s$ is persistent from round $0$, there must be a round $t^*$ at which this execution of $\alg$ stabilizes, i.e., $c_t = c_{t-1}$ for all times $t > t^*$.
    Let $c^*$ be this uniform stable color.
    Then the stable count of nodes with color $c$ is
    \begin{equation}
        C(c) = \left\{\begin{array}{cl}
            n & \text{if $c = c^*$}; \\
            0 & \text{otherwise}.
        \end{array}\right.
    \end{equation}
    In either case, the total variation distance between the stable color distribution and the goal distribution is always lower bounded by a constant:
    \begin{equation}
        \frac{1}{2}\sum_{i=1}^\ell \left|\frac{C(i)}{n} - r(s, i)\right|
        \geq \frac{1}{2} \left|\frac{C(c)}{n} - r(s, c)\right|
        \geq \frac{\min\{r(s, c), 1 - r(s, c)\}}{2}.
    \end{equation}
    So there cannot exist an error term $\varepsilon(n)$ such that $d_{TV}(C/n, r(s)) < \varepsilon(n)$ and $\lim_{n \to \infty} \varepsilon(n) = 0$, contradicting the supposition that $\alg$ solves this instance by stabilizing in a configuration approximating $r(s)$.
    Since the choice of $r$ was arbitrary, all non-homogeneous instances must be impossible for deterministic algorithms.
\end{proof}

In Appendix~\ref{app:hardness}, we re-prove this result with a slightly less pathological counterexample.

\section{A Deterministic Algorithm for Homogeneous Instances} \label{sec:deterministic-algorithm}

In Section~\ref{sec:hardness}, we showed that deterministic algorithms can only hope to solve homogeneous instances of adaptive self-organization.
In this section, we give a deterministic algorithm that solves exactly these instances (Section~\ref{subsec:algorithm}).
It stabilizes in optimal $\bigO{n}$ time when there is no signal (i.e., when $s = \bot$), stabilizes for any signal $s \neq \bot$ with $\bigO{n}$ disruptions, and uses $\bigO{\log n}$ bits of memory per node (Section~\ref{subsec:analysis}).

\subsection{Algorithm Description} \label{subsec:algorithm}

\newcommand{\signal}{\texttt{signal}}
\newcommand{\algcolor}{\texttt{color}}
\newcommand{\state}{\texttt{state}}
\newcommand{\source}{\texttt{source}}
\newcommand{\maxttl}{\texttt{max-ttl}}
\newcommand{\ttl}{\texttt{ttl}}
\newcommand{\timer}{\texttt{timer}}
\newcommand{\dist}{\texttt{dist}}
\newcommand{\reset}{\textsc{Reset}}
\newcommand{\set}{\textsc{Set}}
\newcommand{\locked}{\textsc{Locked}}
\newcommand{\maxrecv}{\ensuremath{m^*}}
\newcommand{\maxtimer}{\ensuremath{t^*}}
\newcommand{\recvttl}{\ensuremath{\ell^*}}
\newcommand{\recvsignal}{\ensuremath{s^*}}
\newcommand{\msg}[2]{\ensuremath{\texttt{#1}(#2)}}

\begin{table}[t]
    \centering
    \caption{Node variables and initialization in Algorithm~\ref{alg:deterministic-coloring}.}
    \label{tab:algvars}
    \resizebox{\textwidth}{!}{\begin{tabular}{llcl}
        \toprule
        \textbf{Variable} & \textbf{Domain} & \textbf{Init.} & \textbf{Description} \\
        \midrule
        \signal & $\{\bot, 1, \ldots, k\}$ & $\bot$ & The signal a node believes is currently present. \\
        \algcolor & $\{1, \ldots, \ell\}$ & $c_\bot$ & A node's current color (i.e., signal response). \\
        \state & $\{\reset, \set, \locked\}$ & \set & A node's current state. \\
        \source & $\{\true, \false\}$ & \true & \true\ iff a node is the source of a broadcast. \\
        \maxttl & $\{2, \ldots, 2n\}$ & $2$ & The maximum time-to-live for this node's messages. \\
        \ttl & $\{0, \ldots, 2n\}$ & $2$ & The current time-to-live for this node's messages. \\
        \timer & $\{0, \ldots, 2n - 1\}$ & $1$ & Counts up to \maxttl, then becomes \locked. \\
        \bottomrule
    \end{tabular}}
\end{table}

Our algorithm has several interlocking mechanisms carefully designed to address the difficulties of adversarial signal and edge dynamics.
We elucidate each of these in the exposition paragraphs below; a reader who is only interested in the final algorithm can skip directly to Algorithm~\ref{alg:deterministic-coloring}.
Throughout, we use the variables and notation summarized in Table~\ref{tab:algvars}.
A minimal running example can be found in Appendix~\ref{app:runningexample}.

\subparagraph{Na\"ive Flooding for a Single Persistent Signal.}

A na\"ive approach to homogeneous adaptive self-organization would have all witnesses of a signal $s \neq \bot$ continuously send messages indicating this fact, causing nodes that receive these messages to change their \algcolor\ to $c_s$ (i.e., the color with $r(s, c_s) = 1$) and continue to propagate these messages.
By 1-interval connectivity, eventually all nodes will receive this message and adopt $\algcolor = c_s$.
This works as long as there is exactly one persistent signal in the system (i.e., $\exists s \neq \bot$ such that $s_t = s$ for all times $t \geq 0$), even if the specific witnesses of signal $s$ change over time.

\subparagraph{Temporal Distance-based Signal Conflict Resolution.}

Of course, our problem allows the signal to change adversarially.
So suppose there is always some signal $s_t \neq \bot$ at each time $t \geq 0$, but that signal may change over time.
Using the same flooding solution as before, it is now possible for a node to receive two (or more) distinct colors $i \neq j$ in messages from its neighbors.
Which should it adopt and propagate?
Any deterministic conflict resolution scheme based only on the \algcolor\ values (e.g., adopting the lexicographically-first color) is obviously fallible (e.g., if $i < j$, then $i$ will dominate the system even if $c_s = j$).

Our solution is for each node to store a variable $\dist \in \mathbb{N}$ representing its temporal distance from the nearest witness.
All witnesses of the signal $s \neq \bot$ adopt color $c_s$ and distance $\dist = 0$, sending messages that contain this information.
Non-witnesses set $\dist = d^* + 1$---where $d^*$ is the minimum \dist\ value in its closed neighborhood---and adopt the lexicographically-first color from among those with $\dist = d^*$, effectively choosing the color of the ``temporally closest'' witness.
This ensures that messages about current signals---which have witnesses and thus have the smallest \dist\ values in the network---will always eventually overtake older ones.

\subparagraph{Time-to-Live for Resets in the Absence of a Signal.}

Now suppose the adversary may send no signal $s_t = \bot$, restoring its complete, problem-defined capabilities.
Our \dist-based solution may now use unbounded memory: without any witnesses in the system to inject $\dist = 0$ messages, older messages will continue to propagate and increment their \dist\ values ad infinitum.
We solve this by reversing the process: instead of incrementing and forwarding a neighborhood's smallest \dist\ value, nodes decrement and forward their neighborhood's largest time-to-live values, \ttl, that witnesses initialize to some \maxttl.
Any node with $\ttl = 0$ ``resets'', setting $\algcolor = c_\bot$.
Thus, if signals are absent from the system for sufficiently long, there will be no witnesses to send new messages with $\ttl = \maxttl$, so all nodes will eventually reach $\ttl = 0$ and $\algcolor = c_\bot$, as desired.

\subparagraph{Discovering a Sufficient Time-to-Live.}

Witnesses of a signal $s \neq \bot$ must initialize \ttl\ values to a sufficiently large \maxttl\ for their messages to reach all other nodes.
Otherwise, nodes may obtain $\ttl = 0$ and reset to $\algcolor = c_\bot$ even when a signal is present simply because its witnesses are temporally far away.
Setting \maxttl\ to the network's temporal diameter is necessary; using $n = |V|$ is sufficient.
But nodes have no a priori knowledge or estimate of $n$.
Thus, we initialize \maxttl\ to a small constant and devise a mechanism for locally detecting when \maxttl\ is too small.
If ever a node detects this event, it doubles its \maxttl\ value, eventually increasing all nodes' \maxttl\ values via flooding.
Doubling \maxttl\ (as opposed to, e.g., incrementing) yields a fast runtime and low memory footprint.

The main challenge in detecting an insufficient \maxttl\ is locally disambiguating whether a node reaches $\ttl = 0$ because the originating witness's \maxttl\ value was too small or because there is no signal.
Our solution uses three node \state\ values: (1) \reset\ nodes that either do not believe a signal is present in the system or are otherwise not participating in flooding, (2) \set\ nodes that are participating in flooding and have set their \algcolor\ accordingly, and (3) \locked\ nodes that believe their \algcolor\ correctly represents the current signal present in the system.
\set\ nodes use a \timer\ variable to count rounds modulo \maxttl, effectively waiting until the flooding they are participating in reaches its furthest extent, at which point they all become \locked.
If ever a non-\locked\ node receives a message from a \locked\ node, it knows \maxttl\ was not large enough; this recipient thus doubles its \maxttl\ value and initiates a new flooding, now with double the reach.
Otherwise, if no such recipient exists, all nodes must be \locked\ with the desired color.

\subparagraph{Final Details and the Complete Algorithm.}

\begin{algorithm}[tp] 
	\caption{Deterministic Algorithm for Homogeneous Instances}
	\label{alg:deterministic-coloring}
	\begin{algorithmic}[1]
		\Statex \textbf{Sending and Receiving Messages.}

		\If {$\state = \reset$}
            \State Send \msg{reset}{\maxttl}. \label{alg:msg-reset}
		\ElsIf {$\state = \set$}
            \State Send \msg{set}{\signal, \maxttl, \ttl, \timer}. \label{alg:msg-set}
		\ElsIf {$\state = \locked$}
            \State Send \msg{locked}{\signal, \maxttl, \ttl}. \label{alg:msg-locked}
		\EndIf

        \medskip

        \Statex \textbf{State Transitions.}

        \State Let \maxrecv\ be the largest \maxttl\ value in any received message of any type. \label{alg:max-recv}
        \If {$\maxrecv > \maxttl$} \label{alg:cond-max-recv}
            \State $\signal \gets \bot$, $\algcolor \gets c_\bot$, $\state \gets \reset$, $\source \gets \false$, $\maxttl \gets \maxrecv$, $\ttl \gets 0$, $\timer \gets 0$.\label{alg:increase-energy-recv}
		\EndIf
		
		\medskip

        \State Let $L$ be the received \msg{locked}{} messages with \maxttl\ values equal to this node's \maxttl. \label{alg:def-lock-msgs}
        \If {$L \neq \emptyset$ and $\state \neq \locked$} \label{alg:cond-lock-msgs}
            \State $\state \gets \set$, $\source \gets \true$, $\maxttl \gets 2 \cdot \maxttl$, $\ttl \gets \maxttl$, $\timer \gets 0$, $L \gets \emptyset$.\label{alg:increase-energy-locked}
		\EndIf
        
        \medskip
		
        \State Let $S$ be the received \msg{set}{} messages with \maxttl\ values equal to this node's \maxttl. \label{alg:def-set-msgs}
        \If {$S \neq \emptyset$ and $\state \neq \locked$} \label{alg:cond-set-msgs}
			\State $\state \gets \set$. \label{alg:join-set-nodes}
            \State $\timer \gets \max\{\timer, \maxtimer\}$, where \maxtimer\ is the largest \timer\ value in any message in $S$. \label{alg:timer-sync}
		\EndIf
		
        \medskip

        \If {$\state \neq \reset$} \label{alg:timer-increment-cond}
            \State $\timer \gets (\timer + 1) \bmod \maxttl$. \label{alg:timer-increment}
			\If {$\timer = 0$} \label{alg:timer-zero-cond}
                \State $\state \gets \locked$. \label{alg:become-locked}
				\If {witnessing a signal $s \neq \bot$} \label{alg:locked-check-signal}
                    \State $\signal \gets s$. \label{alg:locked-new-signal}
					\State $\algcolor \gets c_s$. \label{alg:locked-new-color}
					\State $\source \gets \true$. \label{alg:locked-signal-true}
                \Else {} $\source \gets \false$. \label{alg:locked-signal-false}
				\EndIf
			\EndIf

            \If {$\source = \false$} \label{alg:update-energy-cond} 
                \State Let \recvttl\ be the largest \ttl\ in any message of $L$ or $S$ or this node's \ttl. \label{alg:choose-max-ttl}
                \State Let \recvsignal\ be the largest \signal\ in any message of $L$ or $S$ or this node's $\signal$ with $\ttl = \recvttl$.
                \State $\ttl \gets \recvttl - 1$. \label{alg:update-energy}
                \If {$\ttl = 0$} \label{alg:no-energy}
                    \State $\signal \gets \bot$. \label{alg:no-energy-signal}
                    \State $\algcolor \gets c_\bot$. \label{alg:no-energy-color}
                    \State $\state \gets \reset$. \label{alg:no-energy-state}
                    \State $\timer \gets 0$. \label{alg:no-energy-timer}
                \Else {}
                    \State $\signal \gets \recvsignal$. \label{alg:update-realsignal}
                    \State $\algcolor \gets c_{s^*}$. \label{alg:update-color}
                \EndIf
            \Else {} $\ttl \gets \maxttl$. \label{alg:locked-bcast-energy}
            \EndIf
		\EndIf
		
        \medskip

		\If {$\state = \reset$ and witnessing a signal $s \neq \bot$} \label{alg:become-set-cond}
            \State $\signal \gets s$. \label{alg:become-set-realsignal}
            \State $\algcolor \gets c_s$. \label{alg:become-set-color}
			\State $\state \gets \set$. \label{alg:become-set-state}
			\State $\source \gets \true$. \label{alg:become-set-signal}
			\State $\ttl \gets \maxttl$. \label{alg:become-set-energy}
			\State $\timer \gets 0$. \label{alg:become-set-timer}
		\EndIf
	\end{algorithmic}
\end{algorithm}

Coordinating these mechanisms to achieve our complete solution (Algorithm~\ref{alg:deterministic-coloring}) requires a few more implementation details.
First, the \timer\ variables introduced to control the flooding phases need to be synchronized across nodes even as \maxttl\ is updated dynamically.
To achieve this, \reset\ nodes always maintain $\timer = 0$ and only start counting up after becoming \set.
Before incrementing and forwarding its \timer\ value (Line~\ref{alg:timer-increment}), a \set\ node will first synchronize its timer to the latest \timer\ value from among its own \timer\ and those it receives from its neighbors (Line~\ref{alg:timer-sync}).
If \maxttl\ is large enough, this will ensure that by the time these nodes become \locked, their timers are and remain synchronized.

A second detail addresses issues arising from the need for each flooding phase to run for its full extent in order to detect insufficient \maxttl, even though the adversary may change or remove the signal that started this flooding partway through.
The fix is straightforward: when a flooding phase starts (i.e., when $\timer = 0$, Line~\ref{alg:timer-zero-cond}) any witness of a signal $s \neq \bot$ sets a boolean variable $\source = \true$ (Line~\ref{alg:locked-signal-true}) while all other nodes set $\source = \false$ (Line~\ref{alg:locked-signal-false}).
Nodes with $\source = \true$ behave like witnesses throughout the entire flooding phase---originating messages with $\ttl = \maxttl$ and $\algcolor = c_s$---even if they stop witnessing the signal, only re-evaluating their $\source$ value at the start of the next phase.
Because adaptive self-organization only requires the network to stabilize on the correct color w.r.t.\ signals (or lack thereof) that are persistent (i.e., unchanging after a certain point in time), this phased behavior only delays correct stabilization by at most $\maxttl \leq 2n$ rounds.

Relatedly, any node that detects insufficient \maxttl\ and doubles it becomes \set\ and sets $\source = \true$ (Line~\ref{alg:increase-energy-locked}), immediately initiating a new flooding phase.
This ensures that if \maxttl\ is still insufficient after the doubling, this will be detected within \maxttl\ rounds.
If doubling nodes instead became \reset, there are certain adversarial dynamics that would cause a linear number of rounds to pass before another doubling event could occur.
This optimization helps achieve the algorithm's linear (worst-case optimal) time complexity.

Finally, each node not only maintains a \algcolor\ variable---its current response to whatever signal $s$ it believes to be in the system---but also a \signal\ variable representing $s$ itself.
For the homogeneous instances we consider in this section, \signal\ and \algcolor\ are redundant; nodes could flood either value and still adopt the corresponding goal color.
Choosing to include \signal\ instead of \algcolor\ in messages (Lines~\ref{alg:msg-set} and~\ref{alg:msg-locked}) allows us to later extend this algorithm to non-homogeneous instances that require nodes to stabilize in distributions over multiple colors, as we do in Section~\ref{sec:randomized-alg}.

\subsection{Algorithm Analysis} \label{subsec:analysis}

\newcommand{\varline}[4]{#1_{#2}^{(\ref{#3})}.#4}

Let $v.\texttt{var}$ be the variable \texttt{var} stored by a node $v \in V$.
For any $t \geq 0$, let $v_t.\texttt{var}$ be the value of $v.\texttt{var}$ at time $t$ (i.e., at the start of round $t$) and let $v_t^{(i)}.\texttt{var}$ be the value of $v.\texttt{var}$ after $v$ executes Line~$i$ of Algorithm~\ref{alg:deterministic-coloring} in round $t$.
With this notation defined, we begin our analysis with a series of invariants about the relationships between variables.

\begin{lemma} \label{lem:maxttl-nondecreasing}
    For all nodes $v \in V$, times $t \geq 0$, future times $t' > t$, and algorithm line numbers $i, j$ with $i < j$, $1 < v_t.\maxttl \leq v_t^{(i)}.\maxttl \leq v_t^{(j)}.\maxttl \leq v_{t'}.\maxttl$.
    Moreover, $v_t.\maxttl = 2^x$ for some $x \in \mathbb{Z}_+$.
\end{lemma}
\begin{proof}
    At initialization, $v_0.\maxttl = 2 > 1$ for all nodes $v \in V$ (see Table~\ref{tab:algvars}).
    While executing Algorithm~\ref{alg:deterministic-coloring}, the only updates a node can make to \maxttl\ increase its value (Lines~\ref{alg:increase-energy-recv} and \ref{alg:increase-energy-locked}).
    Line~\ref{alg:increase-energy-recv} updates $\maxttl$ to a value $w_{t-1}.\maxttl$ for some node $w$, while Line~\ref{alg:increase-energy-locked} doubles $\maxttl$.
    In any case, $\maxttl$ will still be a power of 2.
\end{proof}

\begin{lemma} \label{lem:timer-lt-maxttl}
    For all nodes $v \in V$ and times $t \geq 0$, $0 \leq v_t.\timer < v_t.\maxttl$.
\end{lemma}
\begin{proof}
    Argue by induction on time $t \geq 0$.
    At initialization, $v_0.\timer = 1 < 2 = v_0.\maxttl$ for all nodes $v \in V$, as desired (see Table~\ref{tab:algvars}).
    Now suppose the lemma holds for any $t \geq 0$ and consider the execution of Algorithm~\ref{alg:deterministic-coloring} by any node $v$ in round $t$.
    We can safely ignore situations in which $v.\timer$ is unchanged or reset to zero, as the induction hypothesis and the fact that $v.\maxttl$ is nondecreasing (Lemma~\ref{lem:maxttl-nondecreasing}) would imply $0 \leq v_{t+1}.\timer \leq v_t.\timer < v_t.\maxttl \leq v_{t+1}.\maxttl$. 

    The first place $v.\timer$ may increase is Line~\ref{alg:timer-sync}, if a neighbor $u$ with $u_t.\maxttl = \varline{v}{t}{alg:timer-sync}{\maxttl}$ sent $v$ a larger \timer\ value.
    W.l.o.g., suppose $u$ sends the largest \timer\ value among such neighbors. 
    Then $\varline{v}{t}{alg:timer-sync}{\timer} = u_t.\timer$, which by the induction hypothesis is strictly less than $u_t.\maxttl = \varline{v}{t}{alg:timer-sync}{\maxttl}$.
    The only other place $v.\timer$ may increase is the increment in Line~\ref{alg:timer-increment}, but then
    \begin{equation}
        0 \leq \varline{v}{t}{alg:timer-increment}{\timer} = \left(\varline{v}{t}{alg:timer-sync}{\timer} + 1\right) \bmod \varline{v}{t}{alg:timer-increment}{\maxttl} < \varline{v}{t}{alg:timer-increment}{\maxttl}.
    \end{equation}
    Thus, in any case we have $0 \leq v_{t+1}.\timer < v_{t+1}.\maxttl$, completing the induction.
\end{proof}

The next five invariants (Lemmas~\ref{lem:zero-energy}--\ref{lem:ttl-updating}) build on each other to tightly characterize all nodes' \ttl\ values.
In particular, we are concerned with nodes that have the largest \maxttl\ value in the network.
For any time $t \geq 0$, let $m^*(t) = \max_{v \in V} v_t.\maxttl$ be this maximum \maxttl\ value and $e^*(t) = \max_{v \in V}\{v_t.\ttl \mid v_t.\maxttl = m^*(t)\}$ be the maximum \ttl\ value among those nodes with maximum \maxttl.
These invariants' proofs are somewhat tedious---carefully stepping through all possible executions of Algorithm~\ref{alg:deterministic-coloring} in each round depending on different node properties---so we include the complete proof of Lemma~\ref{lem:zero-energy} as a representative example and relegate the rest to Appendix~\ref{app:proofs}.

\begin{lemma} \label{lem:zero-energy}
    For all nodes $v \in V$ and times $t \geq 0$, $v_t.\state = \reset$ if and only if $v_t.\ttl = 0$.
\end{lemma}
\begin{proof}
    Argue by induction on time $t \geq 0$.
    At initialization, $v_0.\state = \set \neq \reset$ and $v_0.\ttl = 2 \neq 0$ for all nodes $v \in V$, as desired (see Table~\ref{tab:algvars}).
    Now suppose the lemma holds for any $t \geq 0$ and consider the execution of Algorithm~\ref{alg:deterministic-coloring} by any node $v$ in round $t$.

    First suppose $\varline{v}{t}{alg:timer-sync}{\state} = \reset$.
    Tracing the rest of Algorithm~\ref{alg:deterministic-coloring} from this value, only Lines~\ref{alg:become-set-state} and~\ref{alg:become-set-energy} may further update $v.\state$ and $v.\ttl$, respectively.
    These updates either occur together or not at all.
    If they both occur, $v_{t+1}.\state = \set$ and $v_{t+1}.\ttl = v_{t+1}.\maxttl > 0$ (by Lemma~\ref{lem:maxttl-nondecreasing}).
    Otherwise, it suffices to show that $\varline{v}{t}{alg:timer-sync}{\ttl} = 0$, yielding $v_{t+1}.\state = \varline{v}{t}{alg:timer-sync}{\state} = \reset$ and $v_{t+1}.\ttl = \varline{v}{t}{alg:timer-sync}{\ttl} = 0$.
    Since $\varline{v}{t}{alg:timer-sync}{\state} = \reset$, $v$ did not execute Lines~\ref{alg:increase-energy-locked} or~\ref{alg:join-set-nodes}.
    It may have executed Line~\ref{alg:increase-energy-recv}, yielding $\varline{v}{t}{alg:increase-energy-recv}{\ttl} = \varline{v}{t}{alg:timer-sync}{\ttl} = 0$, as desired.
    The only remaining possibility is that $v$ has not updated any of its variables in round $t$ by Line~\ref{alg:timer-sync}.
    So $v_t.\state = \varline{v}{t}{alg:timer-sync}{\state} = \reset$ which, by the induction hypothesis, implies $v_t.\ttl = \varline{v}{t}{alg:timer-sync}{\ttl} = 0$.

    Next, suppose $\varline{v}{t}{alg:timer-sync}{\state} \neq \reset$ and $\varline{v}{t}{alg:locked-signal-false}{\source} = \true$, satisfying the condition on Line~\ref{alg:timer-increment-cond} but not the one on Line~\ref{alg:update-energy-cond}, respectively.
    Then $\varline{v}{t}{alg:locked-bcast-energy}{\state} = \varline{v}{t}{alg:become-locked}{\state} \neq \reset$ and $\varline{v}{t}{alg:locked-bcast-energy}{\ttl} = \varline{v}{t}{alg:locked-bcast-energy}{\maxttl} > 0$ (by Lemma~\ref{lem:maxttl-nondecreasing}).
    These values do not satisfy the condition on Line~\ref{alg:become-set-cond}, so we must have $v_{t+1}.\state \neq \reset$ and $v_{t+1}.\ttl \neq 0$.

    Finally, suppose $\varline{v}{t}{alg:timer-sync}{\state} \neq \reset$ and $\varline{v}{t}{alg:locked-signal-false}{\source} = \false$, satisfying the conditions on Lines~\ref{alg:timer-increment-cond} and~\ref{alg:update-energy-cond}, respectively.
    If $\varline{v}{t}{alg:update-energy}{\ttl} = 0$, then $\varline{v}{t}{alg:no-energy-state}{\state} = \reset$.
    Tracing the rest of Algorithm~\ref{alg:deterministic-coloring} from these values, only Lines~\ref{alg:become-set-state} and~\ref{alg:become-set-energy} may further update $v.\state$ and $v.\ttl$; as before, these either occur together or not at all and in either case yield $v_{t+1}.\state = \reset$ if and only if $v_{t+1}.\ttl = 0$.
    Otherwise, if $\varline{v}{t}{alg:update-energy}{\ttl} \neq 0$, then $\varline{v}{t}{alg:no-energy-state}{\state} \neq \reset$.
    Tracing the rest of Algorithm~\ref{alg:deterministic-coloring} from these values, $v.\state$ can never become \reset\ and $v.\ttl$ will never decrease, so we must have $v_{t+1}.\state \neq \reset$ and $v_{t+1}.\ttl \neq 0$.
\end{proof}

\begin{restatable}{lemma}{lemttlbound} \label{lem:ttl-bound}
    For all nodes $v \in V$ and times $t \geq 0$, $0 \leq v_t.\ttl \leq v_t.\maxttl$.
\end{restatable}

\begin{restatable}{lemma}{lemsourcettl} \label{lem:source-ttl}
    For all nodes $v \in V$ and times $t \geq 0$, $v_t.\source = \true$ if and only if $v_t.\ttl = v_t.\maxttl$.
\end{restatable}

\begin{restatable}{lemma}{lemttlvalue} \label{lem:ttl-value}
    Let $M_t[v] = \{u \in N_t(v) \cup \{v\} : u_t.\maxttl = v_{t+1}.\maxttl\}$ be the closed neighborhood of a node $v \in V$ at time $t \geq 0$, restricted to nodes whose \maxttl\ values at time $t$ are the same as $v$'s at time $t+1$.
    Then for all nodes $v \in V$ and times $t \geq 0$,
    \begin{equation} \label{eq:ttl-value}
        v_t.\ttl = \begin{cases}
            2 & \text{if $t = 0$}; \\
            v_t.\maxttl & \text{if $v_t.\source = \true$}; \\
            \max\left\{0, \underset{u \in M_{t-1}[v]}{\max} u_{t-1}.\ttl - 1\right\} & \text{otherwise}.
        \end{cases}
    \end{equation}
\end{restatable}

\begin{restatable}{lemma}{lemttlupdating} \label{lem:ttl-updating}
    For all times $t \geq 0$, if $v_t.\maxttl = m^*(t) = m^*(t+1)$ and $v_t.\source = v_{t+1}.\source = \false$ for all nodes $v \in V$, then $e^*(t+1) = \max\{0, e^*(t) - 1\}$.
\end{restatable}

With our invariants in place, we now argue that Algorithm~\ref{alg:deterministic-coloring} stabilizes with the appropriate node \algcolor\ for a given (persistent) signal.
We first consider the absent signal $s = \bot$ and then turn to signals $s \neq \bot$, deferring the full proofs to Appendix~\ref{app:proofs}.
In this and the remaining results, we slightly abuse the notation $[a, b]$ to mean the set of integers $\{a, a + 1, \ldots, b\}$ and $[a, b)$ to mean $\{a, a + 1, \ldots, b - 1\}$.

\begin{restatable}{lemma}{lempersistnosignal} \label{lem:persist-no-signal}
    Suppose the signal $s_{t_0} = \bot$ is persistent (i.e., there is no signal for all times $t \geq t_0$) and let $m = m^*(t_0)$.
    If $m^*(t) = m$ for all $t \in [t_0, t_0 + 3m]$, then $v_{t'}.\algcolor = c_\bot$ for all nodes $v \in V$ and times $t' \geq t_0 + 3m$.
\end{restatable}
\begin{proof}[Proof Idea]
    Let $t_1 = \min\{t \geq 0 : m^*(t) = m\}$ be the earliest time at which any node obtains $\maxttl = m$; note that this implies $t_1 \leq t_0$.
    We first extend our supposition that $m^*(t) = m$ for all $t \in [t_0, t_0 + 3m]$ to all $t \in [t_1, t_0 + 3m]$, i.e., at least one node has $\maxttl = m$ in this entire interval.
    We then argue that within $m$ rounds of time $t_1$, all nodes have ``caught up'' with $\maxttl = m$.
    Within $2m$ rounds of the signal disappearing from the system, there will not be any nodes originating new broadcasts.
    This causes all nodes' \ttl\ values to decrease by Lemma~\ref{lem:ttl-updating} such that within another $m$ rounds, all nodes will be in the \reset\ state, which we argue will always correspond to having $\algcolor = c_\bot$, as desired.
\end{proof}

\begin{restatable}{lemma}{lempersistsetcolor} \label{lem:persist-set-color}
    Suppose a signal $s_{t_0} \neq \bot$ is persistent (i.e., for all times $t \geq t_0$, some node witnesses $s_t = s_{t_0}$) and let $m = m^*(t_0)$.
    If $m^*(t) = m$ for all $t \geq t_0$, then for all nodes $v \in V$ and times $t \geq t_0 + 4m$, $v_t.\algcolor = c_s$.
\end{restatable}
\begin{proof}[Proof Idea]
    This proof is similar in approach to that of Lemma~\ref{lem:persist-no-signal}.
    We again let $t_1 = \min\{t \geq 0 : m^*(t) = m\}$ be the earliest time at which any node obtains $\maxttl = m$, noting that this implies $t_1 \leq t_0$.
    We then extend our supposition that $m^*(t) = m$ for all $t \geq t_0$ to all $t \geq t_1$, i.e., at least one node has $\maxttl = m$ in this entire interval.
    This allows us to prove two pivotal helper claims: (1) either all nodes have $\state = \locked$ and $\maxttl = m$, or none do, and (2) all nodes with $\state = \locked$ and $\maxttl = m$ have the same \timer\ values.
    We use these facts and our earlier lemmas to build up a long series of claims identifying key points of progress for Algorithm~\ref{alg:deterministic-coloring}.
    A non-exhaustive list---but one that gives the structure of the argument---is as follows: (1) because there is a persistent signal $s \neq \bot$ in the system, there will be some node with $\maxttl = m$ and $\source = \true$ that can start a broadcast with sufficient reach to inform all other nodes of the signal's presence; (2) this causes all nodes to eventually reach $\maxttl = m$; (3) this in turn implies that all nodes will become \locked; (4) this will enable nodes to converge on the correct \algcolor; and (5) all of this occurs within $4m$ rounds of the signal's persistent appearance.
    The actual details of this argument are quite complex, owing to the fact that the network may have been in an apparently arbitrary state before the signal became persistent.
    Thus, there is a flavor of self-stabilization in these arguments to show that the network progressively recovers from the past and converges to a configuration reflecting the true persistent signal.
\end{proof}

Before we can use the previous two lemmas to prove our main result, we must show that nodes' \maxttl\ values cannot increase forever, i.e., that eventually $m^*(t)$ reaches and remains at some maximum value.
The full proof is in Appendix~\ref{app:proofs}.

\begin{restatable}{lemma}{lemmaxttlupperbound} \label{lem:maxttl-upper-bound}
    For all times $t \geq 0$, $m^*(t) \leq 2n$.
\end{restatable}
\begin{proof}[Proof Idea]
    We argue by contradiction, observing that if there is a time $t \geq 0$ with $m^*(t) > 2n$, then this occurred because some node $v \in V$ has $\varline{v}{t-1}{alg:increase-energy-recv}{\maxttl} = m^*(t) / 2$ and then executes Line~\ref{alg:increase-energy-locked} in round $t - 1$.
    The conditions on Line~\ref{alg:cond-lock-msgs} show that this can only occur if $v$ receives a message in round $t-1$ from a neighbor $u \in N_{t-1}(v)$ with $u_{t-1}.\state = \locked$ and $u_{t-1}.\maxttl = m^*(t) / 2$ and either $v_{t-1}.\state = \reset$ or $v_{t-1}.\maxttl < m^*(t) / 2$.
    To contradict this, we first identify the earliest time $t' \leq t$ at which any node obtains $\maxttl = m^*(t) / 2$.
    We then construct a complex induction arguing that at every time $t'' \geq t$, one of two sets of conditions must hold.
    Taken together, these conditions imply that $v_{t-1}.\state = \locked$ and $v_{t-1}.\maxttl = m^*(t) / 2$, so $v$ cannot have executed Line~\ref{alg:increase-energy-locked}.
\end{proof}


With Lemmas~\ref{lem:persist-no-signal}--\ref{lem:maxttl-upper-bound} in hand, we are ready to prove our main result.

\begin{theorem} \label{thm:deterministic}
    Algorithm~\ref{alg:deterministic-coloring} solves any homogeneous instance of adaptive self-organization without nodes having knowledge of $n = |V|$.
    It stabilizes for $s = \bot$ in at most $12n$ rounds, stabilizes for $s \neq \bot$ with at most $16n$ disruptions, and uses $\bigO{\log n}$ bits of memory per node.
\end{theorem}
\begin{proof}
    We first show that Algorithm~\ref{alg:deterministic-coloring} stabilizes to the correct color in linear time for $s = \bot$, i.e., when there is no signal in the network.
    Suppose that $s = \bot$ is persistent from time $t_0 \geq 0$.
    By Lemma~\ref{lem:persist-no-signal}, either the algorithm stabilizes with all nodes having $\algcolor = c_\bot$ within $3m^*(t_0)$ rounds or there is a time $t_1 \in [t_0+1, t_0 + 3m^*(t_0)]$ such that $m^*(t_1) > m^*(t_0)$ by Lemma~\ref{lem:maxttl-nondecreasing}.
    Since $t_1 \geq t_0$, we can apply these lemmas again using $t_1$, and so on to obtain a sequence $t_0 < t_1 < \cdots$ where $t_{i+1} \leq t_i + 3m^*(t_i)$ by Lemma~\ref{lem:persist-no-signal} and $m^*(t_{i+1}) > m^*(t_i)$ by Lemma~\ref{lem:maxttl-nondecreasing}.
    Also by Lemma~\ref{lem:maxttl-nondecreasing}, there exist exponents $x_i \in \mathbb{Z}_+$ such that $m^*(t_i) = 2^{x_i}$ for all $i \geq 0$.
    But Lemma~\ref{lem:maxttl-upper-bound} shows this sequence is finite; i.e., it has a final term $t_r$ with $m^*(t_r) = 2^{x_r} \leq 2n$.
    Thus, the algorithm stabilizes for $s = \bot$ by time
    \begin{equation}
        t_r + 3m^*(t_r) \leq t_0 + \sum_{i=0}^r 3m^*(t_i)
        = t_0 + 3\sum_{i=0}^r 2^{x_i}
        \leq t_0 + 3 \cdot 2^{x_r + 1}
        \leq t_0 + 12n.
    \end{equation}

    Next, we show that Algorithm~\ref{alg:deterministic-coloring} stabilizes with at most a linear number of disruptions for any signal $s \neq \bot$.
    Suppose that the signal $s \neq \bot$ is persistent from time $t_0 \geq 0$.
    Let $t_0 < t_1 < \cdots$ be the sequence of times such that, for all $i \geq 0$, $m^*(t)$ is constant for all $t \in [t_i, t_{i+1})$ and $m^*(t_{i+1}) > m^*(t_i)$; again, the increase follows from Lemma~\ref{lem:maxttl-nondecreasing}.
    Just as in the $s = \bot$ case, Lemma~\ref{lem:maxttl-nondecreasing} gives exponents $x_i \in \mathbb{Z}_+$ such that $m^*(t_i) = 2^{x_i}$, and Lemma~\ref{lem:maxttl-upper-bound} implies a final term $t_r$ in this sequence with $m^*(t_r) = 2^{x_r} \leq 2n$.
    By Lemma~\ref{lem:persist-set-color}, the dynamic network is stable with all nodes having $\algcolor = c_s$ in rounds $t \in [t_i + 4m^*(t_i), t_{i+1})$ for all $i \in [0, r)$ and times $t \geq t_r + 4m^*(t_r)$.
    So there are at most $4m^*(t_i)$ rounds for each $i \in [0, r]$ in which the network is not stable.
    We conclude that the algorithm stabilizes for $s \neq \bot$ with at most $T(n)$ disruptions, where
    \begin{equation}
        T(n) \leq \sum_{i = 0}^r 4m^*(t_i)
        = 4\sum_{i=0}^r 2^{x_i}
        \leq 4 \cdot 2^{x_r + 1}
        \leq 16n.
    \end{equation}

    Finally, we will show that Algorithm~\ref{alg:deterministic-coloring} has $\bigO{\log n}$ space complexity.
    The algorithm maintains seven variables, listed in Table~\ref{tab:algvars}.
    Recall that we assume the number of signals $k$ and colors $\ell$ are fixed constants independent of $n$; thus, $\signal$ and $\algcolor$ use $\bigO{1}$ memory.
    It is easy to see that $\state$ and $\source$ also require only $\bigO{1}$ memory.
    By Lemma~\ref{lem:maxttl-upper-bound}, $\maxttl \leq 2n$ so $\maxttl$ requires no more than $\lceil \log_2{2n} \rceil$ bits to store.
    Finally, by Lemmas~\ref{lem:timer-lt-maxttl} and~\ref{lem:ttl-bound}, $\timer$ and $\ttl$ are both at most $\maxttl$, respectively, also requiring at most $\lceil \log_2{2n} \rceil$ bits.
    Altogether, the algorithm requires $\bigO{\log n}$ bits of memory.
\end{proof}

Theorem~\ref{thm:deterministic} reveals the difference of stabilizing with and without disruptions in Algorithm~\ref{alg:deterministic-coloring}'s response to present vs.\ absent signals, respectively.
Note that \maxttl\ values act as nodes' estimates of the eccentricity of the witnesses, i.e., the number of rounds required for any witness to broadcast a perceived signal to the rest of the network.
Disruptions for present signals $s \neq \bot$ arise from the fact that the dynamics adversary can keep eccentricities---and, by extension, \maxttl\ values---small for arbitrarily long periods of time.
If the adversary suddenly increases witness eccentricities, some node that is ``far'' from a witness may now have its \ttl\ expire, causing it to incorrectly reset to color $c_\perp$ despite a signal being present in the network.
This isn't a correctness issue---eventually, this node will double its \maxttl\ value until the network once again stabilizes at the correct color---but it creates intermediate periods where stabilization is disrupted, requiring restabilization.
If nodes had knowledge of $n = |V|$, this situation could be avoided entirely by initializing $\maxttl = n$ (or the first power of 2 at least $n$), which always upper bounds eccentricity.
By the argument of Lemma~\ref{lem:maxttl-upper-bound}, such a \maxttl\ value will never increase and thus $m^*(t) = n$ for all times $t \geq 0$, allowing a single application of Lemma~\ref{lem:persist-set-color} to imply stabilization in linear time for signals $s \neq \bot$ in $4n$ rounds.

\begin{corollary} \label{cor:knowledgen}
    If nodes have knowledge of $n = |V|$, there is a deterministic algorithm that solves any homogeneous instance of adaptive self-organization using $\bigO{\log n}$ memory and stabilizes in $\bigO{n}$ rounds for all signals $s \in \{\bot, 1, \ldots, k\}$.
\end{corollary}

\section{A Randomized Algorithm for Arbitrary Instances}
\label{sec:randomized-alg}

\begin{algorithm}[t] 
	\caption{Randomized Algorithm for Arbitrary Instances}
	\label{alg:randomized}
	\begin{algorithmic}[1]
        \State Run Algorithm~\ref{alg:deterministic-coloring}, replacing:
        \Statex \quad ``reset'' color updates on Lines~\ref{alg:increase-energy-recv} and \ref{alg:no-energy-color} with $\algcolor \gets \;\sim r(\bot)$.
        \Comment{$\sim$ is sampling.}
        \Statex \quad message-based color updates on Line~\ref{alg:update-color} with $\algcolor \gets \;\sim r(\signal)$.
        \Statex \quad witness color updates w.r.t.\ a signal $s \neq \bot$ on Lines~\ref{alg:locked-new-color} and \ref{alg:become-set-color} with $\algcolor \gets \;\sim r(\signal)$.
	\end{algorithmic}
\end{algorithm}

In Sections~\ref{sec:hardness} and~\ref{sec:deterministic-algorithm}, we proved deterministic algorithms can only solve homogeneous instances of adaptive self-organization and then gave a deterministic algorithm for exactly these instances.
In this section, we give a randomized algorithm that solves any instance of adaptive self-organization w.h.p., homogeneous or otherwise.
This randomized algorithm (Algorithm~\ref{alg:randomized}) is a simple extension of our deterministic one (Algorithm~\ref{alg:deterministic-coloring}): whenever a node updates its \signal, it now samples a new \algcolor\ from the goal distribution $r(\signal)$.
Note that \algcolor\ does not affect the other variables internal to a node, nor is it sent in any messages.
As a result, all other node variables update in this randomized version of the algorithm exactly as they did in the deterministic algorithm.
In particular, once the adversarial signal is stable, all nodes will eventually converge to the correct \signal\ value, then sample their \algcolor\ accordingly.	
That \algcolor\ does not affect the operation of the algorithm also implies that nodes independently sample \algcolor\ values, even though the adversary can adapt to the random choices.
This sampling strategy yields the following lemma.

\begin{lemma} \label{lem:approximation}
    If an execution of Algorithm~\ref{alg:randomized} stabilizes w.r.t.\ a persistent signal $s$, it does so in a configuration approximating $r(s)$ w.h.p.
\end{lemma}
\begin{proof}
    W.l.o.g., let $c(v)$ be the color of node $v$ after stabilization and $C(i)$ be the number of nodes that stabilize with color $i$.
    Define indicator variables $X_v^i$ for nodes $v \in V$ and colors $i \in \{1, \ldots, \ell\}$ that are equal to $1$ if $c(v) = i$ and $0$ otherwise.
    Since nodes sample their colors independently from the goal distribution $r(s)$, it is easy to see that $\E{X_v^i} = \Pr{X_v^i = 1} = r(s, i)$.
    Thus, linearity of expectation implies that the expected frequency of any color $i$ is
    \begin{equation} \label{eq:linexpect}
        \E{\frac{C(i)}{n}} = \E{\frac{1}{n}\sum_{v \in V}X_v^i} = \frac{1}{n}\sum_{v \in V}\E{X_v^i} = r(s, i).
    \end{equation}
    By Hoeffding's inequality~\cite{Hoeffding1963-probabilityinequalities}, the probability that the frequency of color $i$ deviates from its expected value (\eqtext~\ref{eq:linexpect}) by at least some error term $\delta > 0$ is
    \begin{equation} \label{eq:hoeffding}
        \Pr{\left|\frac{C(i)}{n} - r(s, i)\right| \geq \delta} \leq 2\exp(-2n\delta^2).
    \end{equation}
    By the union bound, the probability that every color's frequency has small deviation from $r(s, i)$ is
    \begin{equation} \label{eq:union}
        \Pr{\forall i \in \{1, \ldots, \ell\} : \left|\frac{C(i)}{n} - r(s, i)\right| < \delta} \geq 1 - \sum_{i=1}^\ell \Pr{\left|\frac{C(i)}{n} - r(s, i)\right| \geq \delta}.
    \end{equation}
    If all $\ell$ colors' frequencies deviate from their respective expected values by less than $\delta$, then the total variation distance between the frequency distribution $C$ and the goal distribution $r(s)$ is less than $\ell\delta/2$.
    Combining this with the bounds from \eqstext~\ref{eq:hoeffding} and~\ref{eq:union} yields:
    \begin{align} \label{eq:dtv}
        \Pr{d_{TV}(C, r(s)) < \frac{\ell\delta}{2}} &= \Pr{\frac{1}{2}\sum_{i=1}^\ell \left|\frac{C(i)}{n} - r(s, i)\right| < \frac{\ell\delta}{2}} \\
        &\geq \Pr{\forall i \in \{1, \ldots, \ell\} : \left|\frac{C(i)}{n} - r(s, i)\right| < \delta} \\
        &\geq 1 - 2\ell\exp(-2n\delta^2).
    \end{align}
    Setting $\delta = 2/(\ell n^{1/4})$ yields an error term $\varepsilon(n) = \ell\delta/2 = n^{-1/4}$ which satisfies the definition of approximation given in Section~\ref{subsec:problem} and upper bounds the failure probability by $2\ell\exp(-8\sqrt{n}/\ell^2)$ which achieves the w.h.p.\ result.
\end{proof}

Theorem~\ref{thm:deterministic} and Lemma~\ref{lem:approximation} together imply the following result.

\begin{corollary} \label{cor:randomized}
    Algorithm~\ref{alg:randomized} solves any instance of adaptive self-organization, w.h.p.
    It stabilizes for $s = \bot$ in at most $12n$ rounds, stabilizes for signals $s \neq \bot$ with at most $16n$ disruptions, and uses $\bigO{\log n}$ bits of memory per node.
\end{corollary}

\section{Conclusion} \label{sec:conclude}

We introduced the adaptive self-organization problem for dynamic networks in which nodes must coordinate to self-organize their colors (i.e., states or behaviors) in response to time-varying, locally-perceived environmental signals.
Both the signals and the network topology are allowed to change adversarially.
In this setting, we proved that no deterministic algorithm can solve non-homogeneous instances of this problem, i.e., those in which multiple node colors must persist in the system.
We then designed and analyzed a deterministic algorithm for homogeneous instances that stabilizes in worst-case optimal $\bigO{n}$ time and uses $\bigO{\log n}$ memory per node.
We concluded with two extensions of this algorithm: first, if nodes know $n$, they can achieve stronger stabilization properties; second, if randomization is allowed, then arbitrary instances of adaptive self-organization are solvable, w.h.p., using the same technique and in the same time and space complexity.

This work fully characterizes the computability and optimal time complexity of adaptive self-organization in anonymous, synchronous, 1-interval connected dynamic networks.
It remains open whether the problem is possible, deterministically or otherwise, with sublogarithmic memory per node.
As discussed in Section~\ref{subsec:relwork}, a prior randomized algorithm effectively solved a restricted (one-signal, two-color) version of this problem under a sequential, uniform-at-random scheduler and constrained signal and edge dynamics with \textit{constant memory per node}, but at the cost of a $\bigO{n^6\log n}$ expected runtime~\cite{Oh2025-adaptivecollective}.
Whether these results can be translated to the general problem under standard dynamics, or more broadly what the optimal time--space tradeoffs are for this problem, remain open.

\bibliographystyle{abbrvurl}
\bibliography{ref}

\appendix

\section{An Alternative Proof of Non-Homogeneous Hardness} \label{app:hardness}

In Section~\ref{sec:hardness}, we proved that deterministic algorithms can only hope to solve homogeneous instances of adaptive self-organization.
Our proof considered a pathological case: a static, complete graph where every node is a witness, and thus all nodes always share a common color.
Here, we reprove the theorem using a slightly more interesting case where exactly one node witnesses the signal per round.

\thmhardness*
\begin{proof}[Proof (Alternative)]
    Consider any non-homogeneous instance $r$ of adaptive self-organization; w.l.o.g., let signal $s \in \{\bot, 1, \ldots, k\}$ and color $c \in \{1, \ldots, \ell\}$ satisfy $0 < r(s, c) < 1$. 
    Suppose to the contrary that there exists a deterministic algorithm $\alg$ that solves this instance, and consider the execution of $\alg$ on a dynamic network of $n$ nodes in which $G_t = K_n$, $s_t = s$, and $W_t = \{w\}$ for all rounds $t \geq 0$.
    In words, the network has a static, complete topology; signal $s$ is persistent from the very first round; and $s$ is witnessed only by some fixed node $w$.

    Recall from Section~\ref{subsec:model} that nodes change states based on their current state, the (initially empty) set of messages received in the previous round, and the presence and type of signal they witness (if any).
    Since $\alg$ is deterministic and $G_0$ is a complete graph, symmetry demands that in round $t = 0$, all nodes send each other identical messages.
    Each node then uses the $n - 1$ identical messages it receives and its initial state to perform a state transition.
    The witness $w$ may transition to some state $\alpha_0$ while all other nodes must transition to the same state $\beta_0$ (where possibly $\alpha_0 = \beta_0$).
    In the next round, again since $\alg$ is deterministic and $G_1$ is a complete graph, the non-witness nodes must all receive the same message or no message from $w$ (depending on state $\alpha_0$) and either no messages or $n - 2$ identical messages from the other non-witnesses (depending on $\beta_0$). 
    In any case, all non-witnesses receive the same set of messages, so in the next round they will transition to the same state $\beta_1$ while the witness transitions to a state $\alpha_1$.
    Since the network topology, signal, and witness set never change, this pattern of the witness $w$ being in some state $\alpha_t$ while all other nodes share a state $\beta_t$ must continue for all rounds $t \geq 0$.
    Furthermore, since color is a component of state, there can be at most two colors present in the dynamic network in any round $t$: one color for the witness, and one color shared by all non-witnesses.

    Since $\alg$ solves this instance and signal $s$ is persistent from round $0$, there must be a round $t^*$ at which this execution of $\alg$ stabilizes, i.e., $c_t(v) = c_{t-1}(v)$ for all times $t > t^*$ and nodes $v \in V$. 
    Let $c_w$ be the stable color of witness $w$ and $c_v$ be the stable color shared by all non-witnesses $v \neq w$.
    Then the stable count of nodes with color $c$ is
    \begin{equation}
        C(c) = \left\{\begin{array}{cl}
            0 & \text{if $c \not\in \{c_w, c_v\}$}; \\
            1 & \text{if $c = c_w \neq c_v$}; \\
            n - 1 & \text{if $c = c_v \neq c_w$}; \\
            n & \text{if $c = c_w = c_v$}.
        \end{array}\right.
    \end{equation}
    In any case, the limit of the total variation distance between the stable color distribution and the goal distribution for large $n$ is always lower bounded a constant:
    \begin{align}
        \lim_{n \to \infty} d_{TV}(C, r(s)) &= \lim_{n \to \infty} \frac{1}{2}\sum_{i=1}^\ell \left|\frac{C(i)}{n} - r(s, i)\right| \\
        &\geq \lim_{n \to \infty} \frac{1}{2} \left|\frac{C(c)}{n} - r(s, c)\right| \\
        &\geq \frac{1}{2}\min\{r(s, c), 1 - r(s, c)\}.
    \end{align}
    So there cannot exist an error term $\varepsilon(n)$ such that $d_{TV}(C, r(s)) < \varepsilon(n)$ and $\lim_{n \to \infty} \varepsilon(n) = 0$, contradicting the supposition that $\alg$ solves this instance by stabilizing in a configuration approximating $r(s)$.
    Since the choice of $r$ was arbitrary, all non-homogeneous instances must be impossible for deterministic algorithms.
\end{proof}

\section{Running Example} \label{app:runningexample}

\begin{figure}[tp]
    \centering
    \begin{subfigure}{0.24\textwidth}
        \centering
        \begin{tikzpicture}
            \tikzset{node/.style={circle,draw,node font=\footnotesize, inner sep=0pt, minimum size=0.6cm}}
            \tikzset{src/.style={node, minimum size=0.5cm}}
            \tikzset{label/.style={node distance=0.06cm, node font=\footnotesize, inner sep=0pt}}
            \tikzset{signal/.style={-Latex,thick,decorate,decoration={snake,post length=0.2cm}}};

            \node[node] (1) at (0, 0) {$S$}; \node[src] at (0, 0) {};
            \node[node] (2) at (1.5, 1.3) {$S$}; \node[src] at (1.5, 1.3) {};
            \node[node] (3) at (1.5, 0) {$S$}; \node[src] at (1.5, 0) {};
            \node[node] (4) at (1.5, -1.3) {$S$}; \node[src] at (1.5, -1.3) {};

            \foreach \x in {1, 2, 3, 4}
                \node[label, above=of \x] {$2/2/1$};

            \draw (1) to (2);
            \draw (1) to (3);
            \draw (1) to (4);

            \draw[signal,red!70!black] (-1.1, 0) to (1);
        \end{tikzpicture}
        \caption{\centering $t = 0$, initialization}
    \end{subfigure}
    \hfill
    \begin{subfigure}{0.24\textwidth}
        \centering
        \begin{tikzpicture}
            \tikzset{node/.style={circle,draw,node font=\footnotesize, inner sep=0pt, minimum size=0.6cm}}
            \tikzset{src/.style={node, minimum size=0.5cm}}
            \tikzset{label/.style={node distance=0.06cm, node font=\footnotesize, inner sep=0pt}}
            \tikzset{signal/.style={-Latex,thick,red!70!black,decorate,decoration={snake,post length=0.2cm}}};

            \node[node,fill=red!20] (1) at (0, 0) {$L$}; \node[src] at (0, 0) {};
            \node[node] (2) at (1.5, 1.3) {$L$};
            \node[node] (3) at (1.5, 0) {$L$};
            \node[node] (4) at (1.5, -1.3) {$L$};

            \node[label, above=of 1] {$2/2/0$};
            \foreach \x in {2, 3, 4}
                \node[label, above=of \x] {$2/1/0$};

            \draw (1) to (2);
            \draw (1) to (3);
            \draw (1) to (4);

            \draw[signal,blue!70!black] (-1.1, 0) to (1);
        \end{tikzpicture}
        \caption{\centering $t = 1$}
    \end{subfigure}
    \hfill
    \begin{subfigure}{0.24\textwidth}
        \centering
        \begin{tikzpicture}
            \tikzset{node/.style={circle,draw,node font=\footnotesize, inner sep=0pt, minimum size=0.6cm}}
            \tikzset{src/.style={node, minimum size=0.5cm}}
            \tikzset{label/.style={node distance=0.06cm, node font=\footnotesize, inner sep=0pt}}
            \tikzset{signal/.style={-Latex,thick,red!70!black,decorate,decoration={snake,post length=0.2cm}}};

            \node[node,fill=red!20] (1) at (0, 0) {$L$}; \node[src] at (0, 0) {};
            \node[node,fill=red!20] (2) at (1.5, 1.3) {$L$};
            \node[node,fill=red!20] (3) at (1.5, 0) {$L$};
            \node[node,fill=red!20] (4) at (1.5, -1.3) {$L$};

            \node[label, above=of 1] {$2/2/1$};
            \foreach \x in {2, 3, 4}
                \node[label, above=of \x] {$2/1/1$};

            \draw (1) to (2);
            \draw (1) to (3);
            \draw (1) to (4);

            \draw[signal,blue!70!black] (-1.1, 0) to (1);
        \end{tikzpicture}
        \caption{\centering $t = 2$}
    \end{subfigure}
    \hfill
    \begin{subfigure}{0.24\textwidth}
        \centering
        \begin{tikzpicture}
            \tikzset{node/.style={circle,draw,node font=\footnotesize, inner sep=0pt, minimum size=0.6cm}}
            \tikzset{src/.style={node, minimum size=0.5cm}}
            \tikzset{label/.style={node distance=0.06cm, node font=\footnotesize, inner sep=0pt}}
            \tikzset{signal/.style={-Latex,thick,red!70!black,decorate,decoration={snake,post length=0.2cm}}};

            \node[node,fill=blue!20] (1) at (0, 0) {$L$}; \node[src] at (0, 0) {};
            \node[node,fill=red!20] (2) at (1.5, 1.3) {$L$};
            \node[node,fill=red!20] (3) at (1.5, 0) {$L$};
            \node[node,fill=red!20] (4) at (1.5, -1.3) {$L$};

            \node[label, above=of 1] {$2/2/0$};
            \foreach \x in {2, 3, 4}
                \node[label, above=of \x] {$2/1/0$};

            \draw (1) to (2);
            \draw (1) to (3);
            \draw (1) to (4);

            \draw[signal,blue!70!black] (-1.1, 0) to (1);
        \end{tikzpicture}
        \caption{\centering $t = 3$}
    \end{subfigure} \\ \bigskip
    \begin{subfigure}{0.24\textwidth}
        \centering
        \begin{tikzpicture}
            \tikzset{node/.style={circle,draw,node font=\footnotesize, inner sep=0pt, minimum size=0.6cm}}
            \tikzset{src/.style={node, minimum size=0.5cm}}
            \tikzset{label/.style={node distance=0.06cm, node font=\footnotesize, inner sep=0pt}}
            \tikzset{signal/.style={-Latex,thick,red!70!black,decorate,decoration={snake,post length=0.2cm}}};

            \node[node,fill=blue!20] (1) at (0, 0) {$L$}; \node[src] at (0, 0) {};
            \node[node,fill=blue!20] (2) at (1.5, 1.3) {$L$};
            \node[node,fill=blue!20] (3) at (1.5, 0) {$L$};
            \node[node,fill=blue!20] (4) at (1.5, -1.3) {$L$};

            \node[label, above=of 1] {$2/2/1$};
            \foreach \x in {2, 3, 4}
                \node[label, above=of \x] {$2/1/1$};

            \draw (1) to (2);
            \draw (1) to (3);
            \draw (1) to (4);

            \draw[signal,blue!70!black] (-1.1, 0) to (1);
        \end{tikzpicture}
        \caption{\centering $t = 4$}
    \end{subfigure}
    \hfill
    \begin{subfigure}{0.24\textwidth}
        \centering
        \begin{tikzpicture}
            \tikzset{node/.style={circle,draw,node font=\footnotesize, inner sep=0pt, minimum size=0.6cm}}
            \tikzset{src/.style={node, minimum size=0.5cm}}
            \tikzset{label/.style={node distance=0.06cm, node font=\footnotesize, inner sep=0pt}}
            \tikzset{signal/.style={-Latex,thick,red!70!black,decorate,decoration={snake,post length=0.2cm}}};

            \node[node,fill=blue!20] (1) at (0, 0) {$L$}; \node[src] at (0, 0) {};
            \node[node,fill=blue!20] (2) at (1.5, 1.3) {$L$};
            \node[node,fill=blue!20] (3) at (1.5, 0) {$L$};
            \node[node,fill=blue!20] (4) at (1.5, -1.3) {$L$};

            \node[label, above=of 1] {$2/2/0$};
            \foreach \x in {2, 3, 4}
                \node[label, above=of \x] {$2/1/0$};

            \draw (1) to (2);
            \draw (1) to (3);
            \draw (1) to (4);

            \draw[signal,blue!70!black] (-1.1, 0) to (1);
        \end{tikzpicture}
        \caption{\centering $t = 5$}
    \end{subfigure}
    \hfill
    \begin{subfigure}{0.24\textwidth}
        \centering
        \begin{tikzpicture}
            \tikzset{node/.style={circle,draw,node font=\footnotesize, inner sep=0pt, minimum size=0.6cm}}
            \tikzset{src/.style={node, minimum size=0.5cm}}
            \tikzset{label/.style={node distance=0.06cm, node font=\footnotesize, inner sep=0pt}}
            \tikzset{signal/.style={-Latex,thick,red!70!black,decorate,decoration={snake,post length=0.2cm}}};

            \node[node,fill=blue!20] (1) at (0, 0) {$L$}; \node[src] at (0, 0) {};
            \node[node,fill=blue!20] (2) at (1.5, -0.867) {$L$};
            \node[node,fill=blue!20] (3) at (0, -1.73) {$L$};
            \node[node,fill=blue!20] (4) at (1.5, -2.6) {$L$};

            \node[label, above=of 1] {$2/2/1$};
            \foreach \x in {2, 3, 4}
                \node[label, above=of \x] {$2/1/1$};

            \draw (1) to (2) to (3) to (4);

            \draw[signal,blue!70!black] (-1.1, 0) to (1);
        \end{tikzpicture}
        \caption{\centering $t = 6$}
    \end{subfigure}
    \hfill
    \begin{subfigure}{0.24\textwidth}
        \centering
        \begin{tikzpicture}
            \tikzset{node/.style={circle,draw,node font=\footnotesize, inner sep=0pt, minimum size=0.6cm}}
            \tikzset{src/.style={node, minimum size=0.5cm}}
            \tikzset{label/.style={node distance=0.06cm, node font=\footnotesize, inner sep=0pt}}
            \tikzset{signal/.style={-Latex,thick,red!70!black,decorate,decoration={snake,post length=0.2cm}}};

            \node[node,fill=blue!20] (1) at (0, 0) {$L$}; \node[src] at (0, 0) {};
            \node[node,fill=blue!20] (2) at (1.5, -0.867) {$L$};
            \node[node] (3) at (0, -1.73) {$R$};
            \node[node] (4) at (1.5, -2.6) {$R$};

            \node[label, above=of 1] {$2/2/0$};
            \node[label, above=of 2] {$2/1/0$};
            \foreach \x in {3, 4}
                \node[label, above=of \x] {$2/0/0$};

            \draw (1) to (2) to (3) to (4);

            \draw[signal,blue!70!black] (-1.1, 0) to (1);
        \end{tikzpicture}
        \caption{\centering $t = 7$}
    \end{subfigure} \\ \bigskip
    \begin{subfigure}{0.24\textwidth}
        \centering
        \begin{tikzpicture}
            \tikzset{node/.style={circle,draw,node font=\footnotesize, inner sep=0pt, minimum size=0.6cm}}
            \tikzset{src/.style={node, minimum size=0.5cm}}
            \tikzset{label/.style={node distance=0.06cm, node font=\footnotesize, inner sep=0pt}}
            \tikzset{signal/.style={-Latex,thick,red!70!black,decorate,decoration={snake,post length=0.2cm}}};

            \node[node,fill=blue!20] (1) at (0, 0) {$L$}; \node[src] at (0, 0) {};
            \node[node,fill=blue!20] (2) at (1.5, -0.867) {$L$};
            \node[node] (3) at (0, -1.73) {$S$}; \node[src] at (0, -1.73) {};
            \node[node] (4) at (1.5, -2.6) {$R$};

            \node[label, above=of 1] {$2/2/1$};
            \node[label, above=of 2] {$2/1/1$};
            \node[label, above=of 3] {$4/4/1$};
            \node[label, above=of 4] {$2/0/0$};

            \draw (1) to (2) to (3) to (4);

            \draw[signal,blue!70!black] (-1.1, 0) to (1);
        \end{tikzpicture}
        \caption{\centering $t = 8$}
    \end{subfigure}
    \hfill
    \begin{subfigure}{0.24\textwidth}
        \centering
        \begin{tikzpicture}
            \tikzset{node/.style={circle,draw,node font=\footnotesize, inner sep=0pt, minimum size=0.6cm}}
            \tikzset{src/.style={node, minimum size=0.5cm}}
            \tikzset{label/.style={node distance=0.06cm, node font=\footnotesize, inner sep=0pt}}
            \tikzset{signal/.style={-Latex,thick,red!70!black,decorate,decoration={snake,post length=0.2cm}}};

            \node[node] (1) at (0, 0) {$S$};
            \node[node] (2) at (1.5, -0.867) {$S$};
            \node[node] (3) at (0, -1.73) {$S$}; \node[src] at (0, -1.73) {};
            \node[node] (4) at (1.5, -2.6) {$S$};

            \node[label, above=of 1] {$4/2/3$};
            \node[label, above=of 2] {$4/3/3$};
            \node[label, above=of 3] {$4/4/3$};
            \node[label, above=of 4] {$4/3/3$};

            \draw (1) to (2) to (3) to (4);

            \draw[signal,blue!70!black] (-1.1, 0) to (1);
        \end{tikzpicture}
        \caption{\centering $t = 10$}
    \end{subfigure}
    \hfill
    \begin{subfigure}{0.24\textwidth}
        \centering
        \begin{tikzpicture}
            \tikzset{node/.style={circle,draw,node font=\footnotesize, inner sep=0pt, minimum size=0.6cm}}
            \tikzset{src/.style={node, minimum size=0.5cm}}
            \tikzset{label/.style={node distance=0.06cm, node font=\footnotesize, inner sep=0pt}}
            \tikzset{signal/.style={-Latex,thick,red!70!black,decorate,decoration={snake,post length=0.2cm}}};

            \node[node,fill=blue!20] (1) at (0, 0) {$L$}; \node[src] at (0, 0) {}; 
            \node[node] (2) at (1.5, -0.867) {$L$};
            \node[node] (3) at (0, -1.73) {$L$};
            \node[node] (4) at (1.5, -2.6) {$L$};

            \node[label, above=of 1] {$4/4/0$};
            \node[label, above=of 2] {$4/3/0$};
            \node[label, above=of 3] {$4/3/0$};
            \node[label, above=of 4] {$4/3/0$};

            \draw (1) to (2) to (3) to (4);

            \draw[signal,blue!70!black] (-1.1, 0) to (1);
        \end{tikzpicture}
        \caption{\centering $t = 11$}
    \end{subfigure}
    \hfill
    \begin{subfigure}{0.24\textwidth}
        \centering
        \begin{tikzpicture}
            \tikzset{node/.style={circle,draw,node font=\footnotesize, inner sep=0pt, minimum size=0.6cm}}
            \tikzset{src/.style={node, minimum size=0.5cm}}
            \tikzset{label/.style={node distance=0.06cm, node font=\footnotesize, inner sep=0pt}}
            \tikzset{signal/.style={-Latex,thick,red!70!black,decorate,decoration={snake,post length=0.2cm}}};

            \node[node,fill=blue!20] (1) at (0, 0) {$L$}; \node[src] at (0, 0) {}; 
            \node[node,fill=blue!20] (2) at (1.5, -0.867) {$L$};
            \node[node,fill=blue!20] (3) at (0, -1.73) {$L$};
            \node[node,fill=blue!20] (4) at (1.5, -2.6) {$L$};

            \node[label, above=of 1] {$4/4/3$};
            \node[label, above=of 2] {$4/3/3$};
            \node[label, above=of 3] {$4/2/3$};
            \node[label, above=of 4] {$4/1/3$};

            \draw (1) to (2) to (3) to (4);
        \end{tikzpicture}
        \caption{\centering $t = 14$}
    \end{subfigure} \\ \bigskip
    \begin{subfigure}{0.32\textwidth}
        \centering
        \begin{tikzpicture}
            \tikzset{node/.style={circle,draw,node font=\footnotesize, inner sep=0pt, minimum size=0.6cm}}
            \tikzset{src/.style={node, minimum size=0.5cm}}
            \tikzset{label/.style={node distance=0.06cm, node font=\footnotesize, inner sep=0pt}}
            \tikzset{signal/.style={-Latex,thick,red!70!black,decorate,decoration={snake,post length=0.2cm}}};

            \node[node,fill=blue!20] (1) at (0, 0) {$L$};
            \node[node,fill=blue!20] (2) at (1.5, -0.867) {$L$};
            \node[node,fill=blue!20] (3) at (0, -1.73) {$L$};
            \node[node,fill=blue!20] (4) at (1.5, -2.6) {$L$};

            \node[label, above=of 1] {$4/3/0$};
            \node[label, above=of 2] {$4/3/0$};
            \node[label, above=of 3] {$4/2/0$};
            \node[label, above=of 4] {$4/1/0$};

            \draw (1) to (2) to (3) to (4);
        \end{tikzpicture}
        \caption{\centering $t = 15$}
    \end{subfigure}
    \hfill
    \begin{subfigure}{0.32\textwidth}
        \centering
        \begin{tikzpicture}
            \tikzset{node/.style={circle,draw,node font=\footnotesize, inner sep=0pt, minimum size=0.6cm}}
            \tikzset{src/.style={node, minimum size=0.5cm}}
            \tikzset{label/.style={node distance=0.06cm, node font=\footnotesize, inner sep=0pt}}
            \tikzset{signal/.style={-Latex,thick,red!70!black,decorate,decoration={snake,post length=0.2cm}}};

            \node[node,fill=blue!20] (1) at (0, 0) {$L$};
            \node[node,fill=blue!20] (2) at (1.5, -0.867) {$L$};
            \node[node,fill=blue!20] (3) at (0, -1.73) {$L$};
            \node[node,fill=blue!20] (4) at (1.5, -2.6) {$L$};

            \foreach \x in {1, 2, 3, 4}
                \node[label, above=of \x] {$4/1/2$};

            \draw (1) to (2) to (3) to (4);
        \end{tikzpicture}
        \caption{\centering $t = 17$}
    \end{subfigure}
    \hfill
    \begin{subfigure}{0.32\textwidth}
        \centering
        \begin{tikzpicture}
            \tikzset{node/.style={circle,draw,node font=\footnotesize, inner sep=0pt, minimum size=0.6cm}}
            \tikzset{src/.style={node, minimum size=0.5cm}}
            \tikzset{label/.style={node distance=0.06cm, node font=\footnotesize, inner sep=0pt}}
            \tikzset{signal/.style={-Latex,thick,red!70!black,decorate,decoration={snake,post length=0.2cm}}};

            \node[node] (1) at (0, 0) {$R$};
            \node[node] (2) at (1.5, -0.867) {$R$};
            \node[node] (3) at (0, -1.73) {$R$};
            \node[node] (4) at (1.5, -2.6) {$R$};

            \foreach \x in {1, 2, 3, 4}
                \node[label, above=of \x] {$4/0/0$};

            \draw (1) to (2) to (3) to (4);
        \end{tikzpicture}
        \caption{\centering $t = 18$, stable if $s = \bot$ persists}
    \end{subfigure}
    \caption{Running example of Algorithm~\ref{alg:deterministic-coloring} as narrated in Appendix~\ref{app:runningexample}.
    For simplicity, this homogeneous instance of adaptive self-organization has signals $\{\bot, \text{\color{red!70!black} red}, \text{\color{blue!70!black} blue}\}$ mapping to node colors {\tikz \node[circle,draw,inner sep=0pt, minimum size=0.2cm] {};}, {\tikz \node[circle,draw,inner sep=0pt, minimum size=0.2cm,fill=red!20] {};}, and {\tikz \node[circle,draw,inner sep=0pt, minimum size=0.2cm,fill=blue!20] {};}, respectively.
    Node states are shown as in-circle labels: $R$ for \reset, $S$ for \set, and $L$ for \locked.
    Nodes with double circles have $\source = \true$.
    Text annotations near nodes represent their $\maxttl / \ttl / \timer$ values.}
    \label{fig:runningexample}
\end{figure}

\figtext~\ref{fig:runningexample} shows a running example of our deterministic algorithm for homogeneous instances (Algorithm~\ref{alg:deterministic-coloring}).
Here, we consider a two-signal, two-color homogeneous instance.
At initialization ($t = 0$), nodes take on their initial values from Table~\ref{tab:algvars}.
When the nodes' \timer\ values reach $0$ in round $0$, all nodes become \locked.
The left node additionally witnesses the red signal and thus adopts the red color and becomes a \source\ for a red broadcast ($t = 1$). 
The signal adversary immediately replaces the red signal with a blue one, but nodes do not re-check what signals they witness until their \timer\ values reach $0$ again (in round $2$).
In the meantime, the red color is broadcast through the network.
Once the blue signal is witnessed and its broadcast is started ($t = 3$), the blue color propagates through the network ($t = 5$).
Without further topological changes, the dynamic network would remain stable.

However, at $t = 6$, the adversary rearranges the network as a path.
This increases the eccentricity of the witness, resulting in round $6$ when the last two nodes of the path do not receive any blue messages, reach $\timer = 0$, and \reset\ ($t = 7$).
This is an example of a ``disruption'' discussed after Theorem~\ref{thm:deterministic}; the blue signal's presence has not changed, but with increased witness eccentricity, nodes may temporarily diverge from the correct color.
In the next round, one of these \reset\ nodes receives a message from a \locked\ node, detecting that an ongoing broadcast attempt did not reach far enough.
It responds by doubling its \maxttl\ value and initiating a new broadcast (rounds $7$--$9$), during which all other nodes learn of the increased \maxttl\ and synchronize their \timer\ values.

In round 10, the nodes' \timer\ values again reach $0$, allowing the first node to once again witness the blue signal and initiate a new corresponding broadcast, which reaches all other nodes by $t = 14$.
Now that $\maxttl \geq n$, even further topological changes cannot disrupt this stable blue configuration as long as the blue signal persists.
However, at $t = 14$, the signal adversary removes the signal.
Without a source to inject new $\ttl = \maxttl$ messages, nodes' \ttl\ values decrement until reaching $0$, causing all nodes to \reset\ ($t = 18$).

\section{Omitted Proofs} \label{app:proofs}

In this appendix, we present complete proofs that were omitted from the main text for clarity or due to space constraints.

\lemttlbound*
\begin{proof}
    Argue by induction on time $t \geq 0$.
    At initialization, $v_0.\ttl = v_0.\maxttl = 2$ for all nodes $v \in V$, as desired (see Table~\ref{tab:algvars}).
    Now suppose the lemma holds for any $t \geq 0$ and consider the execution of Algorithm~\ref{alg:deterministic-coloring} by any node $v$ in round $t$.

    By the induction hypothesis, node $v$ starts round $t$ with values $0 \leq v_t.\ttl \leq v_t.\maxttl$.
    Node $v$ may first update these values in Line~\ref{alg:increase-energy-recv} to obtain $\varline{v}{t}{alg:increase-energy-recv}{\ttl} = 0$, which is strictly less than $\varline{v}{t}{alg:increase-energy-recv}{\maxttl} = \max_{u \in N_t(v) \cup \{v\}}\{u_t.\maxttl\}$ by Lemma~\ref{lem:maxttl-nondecreasing}.
    If $v$ executes Line~\ref{alg:increase-energy-locked}, then $\varline{v}{t}{alg:increase-energy-locked}{\ttl} = \varline{v}{t}{alg:increase-energy-locked}{\maxttl} = 2 \cdot \max_{u \in N_t(v) \cup \{v\}}\{u_t.\maxttl\} > 0$, again by Lemma~\ref{lem:maxttl-nondecreasing}.
    So in any case, the lemma holds through Line~\ref{alg:increase-energy-locked}.

    Next, $v$ may set $v.\ttl \gets \recvttl - 1$ in Line~\ref{alg:update-energy}, where $\recvttl$ is either $u_t.\ttl$ for some $u \in N_t(v)$ or $\varline{v}{t}{alg:increase-energy-locked}{\ttl}$.
    If the former, $\varline{v}{t}{alg:update-energy}{\ttl} = u_t.\ttl - 1 < u_t.\maxttl$ by the induction hypothesis, and $u_t.\maxttl \leq \varline{v}{t}{alg:update-energy}{\maxttl}$ by Lines~\ref{alg:max-recv}--\ref{alg:increase-energy-recv}.
    If the latter, $\varline{v}{t}{alg:update-energy}{\ttl} = \varline{v}{t}{alg:increase-energy-locked}{\ttl} - 1 < \varline{v}{t}{alg:increase-energy-locked}{\maxttl} = \varline{v}{t}{alg:update-energy}{\maxttl}$ by the analysis of the previous paragraph.
    In either case, $\varline{v}{t}{alg:update-energy}{\ttl} < \varline{v}{t}{alg:update-energy}{\maxttl}$.

    To claim the lemma holds through Line~\ref{alg:update-energy}, it remains to show that $\varline{v}{t}{alg:update-energy}{\ttl} \geq 0$.
    Suppose to the contrary that $\varline{v}{t}{alg:update-energy}{\ttl} = \recvttl - 1 < 0$.
    By definition, $\recvttl \geq \varline{v}{t}{alg:choose-max-ttl}{\ttl} = \varline{v}{t}{alg:increase-energy-locked}{\ttl}$.
    So we must have $\varline{v}{t}{alg:increase-energy-locked}{\ttl} < 1$, and since the lemma holds through Line~\ref{alg:increase-energy-locked}, it must actually be $\varline{v}{t}{alg:increase-energy-locked}{\ttl} = 0$.
    This, in turn, is only possible if either $v$ executed Line~\ref{alg:increase-energy-recv} but not Line~\ref{alg:increase-energy-locked} or $v$ did not update any variables in round $t$ by Line~\ref{alg:increase-energy-locked} and $v_t.\ttl < 1$, which the induction hypothesis simplifies to $v_t.\ttl = 0$.
    Both cases imply $\varline{v}{t}{alg:increase-energy-locked}{\state} = \reset$ (the latter uses Lemma~\ref{lem:zero-energy}).
    But we must have $\varline{v}{t}{alg:timer-sync}{\state} \neq \reset$ in order for $v$ to execute Line~\ref{alg:update-energy}, so $v.\state$ must become \set\ in Line~\ref{alg:join-set-nodes}.
    This can only happen if the condition in Line~\ref{alg:cond-set-msgs} is satisfied, implying $v$ received a message from some $u \in N_t(v)$ with $u_t.\state = \set$ and $u_t.\maxttl = \varline{v}{t}{alg:increase-energy-locked}{\maxttl}$.
    These properties make $u_t.\ttl$ a candidate for $\recvttl$; i.e., $\recvttl \geq u_t.\ttl$.
    The induction hypothesis implies $u_t.\ttl \geq 0$, and Lemma~\ref{lem:zero-energy} implies that $u_t.\ttl \neq 0$ since $u_t.\state \neq \reset$.
    Combining these facts yields $\recvttl \geq u_t.\ttl > 0$, contradicting the assumption that $\recvttl - 1 < 0$.
    Thus, the lemma holds through Line~\ref{alg:update-energy}.

    Finally, $v$ may set $v.\ttl \gets v.\maxttl$ in Lines~\ref{alg:locked-bcast-energy} or~\ref{alg:become-set-energy}.
    These updates clearly maintain the lemma, so we conclude that $0 \leq v_{t+1}.\ttl \leq v_{t+1}.\maxttl$, completing the induction.
\end{proof}

\lemsourcettl*
\begin{proof}
    At initialization, $v_0.\source = \true$ and $v_0.\maxttl = v_0.\ttl = 2$ for all nodes $v \in V$, so the lemma holds (see Table~\ref{tab:algvars}).
    Now suppose the lemma holds for any $t \geq 0$ and consider the execution of Algorithm~\ref{alg:deterministic-coloring} by any node $v$ in round $t$.

    Working backwards, node $v$ executes Line~\ref{alg:become-set-signal} if and only if it executes Line~\ref{alg:become-set-energy}.
    So either $v_{t+1}.\source = \varline{v}{t}{alg:become-set-signal}{\source} = \true$ and $v_{t+1}.\ttl = \varline{v}{t}{alg:become-set-energy}{\ttl} = \varline{v}{t}{alg:become-set-energy}{\maxttl} = v_{t+1}.\maxttl$, completing the induction, or $v_{t+1}.* = \varline{v}{t}{alg:locked-bcast-energy}{*}$.
    Note that $\varline{v}{t}{alg:locked-bcast-energy}{\maxttl} = \varline{v}{t}{alg:increase-energy-locked}{\maxttl}$ regardless of how Lines~\ref{alg:timer-increment-cond}--\ref{alg:locked-bcast-energy} are executed, so we need only focus on updates to $v.\source$ and $v.\ttl$ in these lines.
    If these lines are executed (i.e., if $\varline{v}{t}{alg:timer-sync}{\state} \neq \reset$), then either $\varline{v}{t}{alg:locked-bcast-energy}{\source} = \varline{v}{t}{alg:locked-signal-false}{\source} = \true$ and Line~\ref{alg:locked-bcast-energy} is executed, yielding $\varline{v}{t}{alg:locked-bcast-energy}{\ttl} = \varline{v}{t}{alg:locked-bcast-energy}{\maxttl}$, or $\varline{v}{t}{alg:locked-bcast-energy}{\source} = \varline{v}{t}{alg:locked-signal-false}{\source} = \false$ and Line~\ref{alg:update-energy} is executed, yielding $\varline{v}{t}{alg:locked-bcast-energy}{\ttl} = \varline{v}{t}{alg:update-energy}{\ttl} = \recvttl - 1 < \varline{v}{t}{alg:locked-bcast-energy}{\maxttl}$ (by the same argument as in the proof of Lemma~\ref{lem:ttl-bound}).
    Either way, $\varline{v}{t}{alg:locked-bcast-energy}{\source} = \true$ if and only if $\varline{v}{t}{alg:locked-bcast-energy}{\ttl} = \varline{v}{t}{alg:locked-bcast-energy}{\maxttl}$, completing the induction.

    So backtrack further, supposing Lines~\ref{alg:timer-increment-cond}--\ref{alg:locked-bcast-energy} are not executed (i.e., $\varline{v}{t}{alg:timer-sync}{\state} = \reset$).
    Then $\varline{v}{t}{alg:locked-bcast-energy}{*} = \varline{v}{t}{alg:timer-sync}{*}$.
    In fact, the three variables of interest to this lemma are not modified in Lines~\ref{alg:def-set-msgs}--\ref{alg:timer-sync}, so it suffices to consider their values after Line~\ref{alg:increase-energy-locked}.
    If Line~\ref{alg:increase-energy-locked} is executed, then $\varline{v}{t}{alg:increase-energy-locked}{\source} = \true$ and $\varline{v}{t}{alg:increase-energy-locked}{\ttl} = \varline{v}{t}{alg:increase-energy-locked}{\maxttl}$, completing the induction.
    Otherwise, if Line~\ref{alg:increase-energy-locked} is not executed but Line~\ref{alg:increase-energy-recv} is, then $\varline{v}{t}{alg:increase-energy-locked}{\source} = \false$ and $\varline{v}{t}{alg:increase-energy-locked}{\ttl} = 0 < \varline{v}{t}{alg:increase-energy-locked}{\maxttl}$ by Lemma~\ref{lem:maxttl-nondecreasing}, also completing the induction.
    The only remaining possibility is that $v$ does not update any of $v.\source$, $v.\maxttl$, and $v.\ttl$ in round $t$, but then $v_{t+1}.\source = v_t.\source$, $v_{t+1}.\maxttl = v_t.\maxttl$, and $v_{t+1}.\ttl = v_t.\ttl$, so the induction hypothesis completes the induction.
\end{proof}

\lemttlvalue*
\begin{proof}
    If $t = 0$, then $v_0.\ttl = 2$ for all nodes $v \in V$, as desired (see Table~\ref{tab:algvars}).
    If $v_t.\source = \true$, then $v_t.\ttl = v_t.\maxttl$ by Lemma~\ref{lem:source-ttl}, again as desired.

    So consider any node $v$ and time $t > 0$ such that $v_t.\source = \false$; we will show $v_t.\ttl = \max\{0, \max_{u \in M_{t-1}[v]} u_{t-1}.\ttl - 1\}$.
    Consider the execution of Algorithm~\ref{alg:deterministic-coloring} by node $v$ in round $t - 1$.
    We may ignore the $v.\ttl$ updates in Lines~\ref{alg:locked-bcast-energy} and~\ref{alg:become-set-energy}, as their execution implies $v_t.\source = \true$.

    First suppose $\varline{v}{t-1}{alg:timer-sync}{\state} = \reset$.
    Tracing the rest of Algorithm~\ref{alg:deterministic-coloring} from this value (and recalling we may ignore Line~\ref{alg:become-set-energy}) yields $v_t.\state = \reset$, implying $v_t.\ttl = 0$ by Lemma~\ref{lem:zero-energy}.
    Also, $\varline{v}{t-1}{alg:timer-sync}{\state} = \reset$ implies $v_{t-1}.\state = \reset$ or $v$ executes Line~\ref{alg:increase-energy-recv} in round $t-1$, yielding $v_{t-1}.\maxttl < v_t.\maxttl$ by Lemma~\ref{lem:maxttl-nondecreasing} and thus $v \not\in M_{t-1}[v]$.
    Now consider any neighbor $u \in M_{t-1}[v] \setminus \{v\}$.
    By Lines~\ref{alg:max-recv}--\ref{alg:increase-energy-recv}, we must have $u_{t-1}.\maxttl = \varline{v}{t-1}{alg:increase-energy-recv}{\maxttl}$.
    Then $u_{t-1}.\state \neq \locked$, as this would otherwise satisfy the condition on Line~\ref{alg:cond-lock-msgs} and cause $v$ to become \set, contradicting $\varline{v}{t-1}{alg:timer-sync}{\state} = \reset$.
    An analogous argument for Line~\ref{alg:cond-set-msgs} shows $u_{t-1}.\state  \neq \set$.
    Thus, $u_{t-1}.\state = \reset$ for all $u \in M_{t-1}[v]$.
    By Lemma~\ref{lem:zero-energy}, this implies $u_{t-1}.\ttl = 0$ for all $u \in M_{t-1}[v]$, and thus $\max\{0, \max_{u \in M_{t-1}[v]} u_{t-1}.\ttl - 1\} = 0 = v_t.\ttl$.

    Now suppose that $\varline{v}{t-1}{alg:timer-sync}{\state} \neq \reset$.
    Since $v_t.\source = \false$, $v$ must have $v_t.\ttl = \varline{v}{t-1}{alg:update-energy}{\ttl} = \recvttl - 1$, where $\recvttl$ (defined on Line~\ref{alg:choose-max-ttl}) can be written as
    \begin{equation} \label{eq:?}
        \recvttl = \max\left(\{u_{t-1}.\ttl : u \in M_{t-1}[v] \setminus \{v\}\} \cup \{\varline{v}{t-1}{alg:update-energy-cond}{\ttl}\}\right).
    \end{equation}
    This follows from two facts.
    First, although including neighbors $u \in M_{t-1}[v]$ with $u_{t-1}.\state = \reset$ in this definition diverges from Line~\ref{alg:choose-max-ttl} which only considers \locked\ and \set\ neighbors, it does not change the value of $\recvttl$ since any such neighbors have $u_{t-1}.\ttl = 0$ by Lemma~\ref{lem:zero-energy} but $v_t.\ttl = \recvttl - 1$ implies $\recvttl \geq 1$ by Lemma~\ref{lem:ttl-bound}.
    Second, on Line~\ref{alg:choose-max-ttl}, all messages in $L$ (defined on Line~\ref{alg:def-lock-msgs}) and $S$ (defined on Line~\ref{alg:def-set-msgs}) originate from neighbors $u \in N_{t-1}(v)$ with $u_{t-1}.\maxttl = v_t.\maxttl$ that are \locked\ or \set, respectively.
    It is easy to see this is true for $S$, which is defined after any updates to $v.\maxttl$ in round $t-1$.
    It is analogously easy to see for $L$ if $v$ does not execute Line~\ref{alg:increase-energy-locked}.
    But if Line~\ref{alg:increase-energy-locked} is run, $L$ is reset to $\emptyset$, reflecting the fact that no neighbors $u \in N_{t-1}(v)$ have $u_{t-1}.\maxttl = v_t.\maxttl = 2 \cdot \varline{v}{t-1}{alg:increase-energy-recv}{\maxttl}$.
    If such a neighbor existed, then by Line~\ref{alg:increase-energy-recv}, $\varline{v}{t-1}{alg:increase-energy-recv}{\maxttl} \geq u_{t-1}.\maxttl = 2 \cdot \varline{v}{t-1}{alg:increase-energy-recv}{\maxttl}$, contradicting Lemma~\ref{lem:maxttl-nondecreasing} with $\varline{v}{t-1}{alg:increase-energy-recv}{\maxttl} < 0$.
    
    In any case, the above representation of $\recvttl$ faithfully captures its definition in Line~\ref{alg:choose-max-ttl}.
    It remains to show that $v_{t-1}.\ttl = \varline{v}{t-1}{alg:update-energy-cond}{\ttl}$ if and only if $v \in M_{t-1}[v]$, which will further collapse this representation to the form given in the lemma.
    If $v \in M_{t-1}[v]$ (i.e., $v_{t-1}.\maxttl = v_t.\maxttl$), then $v$ must not execute Lines~\ref{alg:increase-energy-recv} or~\ref{alg:increase-energy-locked} in round $t-1$ and thus $v_{t-1}.\ttl = \varline{v}{t-1}{alg:update-energy-cond}{\ttl}$.
    Otherwise, suppose $v \not\in M_{t-1}[v]$, implying $v$ executes one or both of Lines~\ref{alg:increase-energy-recv} and~\ref{alg:increase-energy-locked} in round $t-1$; we will show that in fact, $v$ only executes Line~\ref{alg:increase-energy-recv}.
    Suppose to the contrary $v$ executes Line~\ref{alg:increase-energy-locked} in round $t-1$, yielding $\varline{v}{t-1}{alg:increase-energy-locked}{\source} = \true$, $\varline{v}{t-1}{alg:increase-energy-locked}{\maxttl} \geq 4$ (by Lemma~\ref{lem:maxttl-nondecreasing}), and $\varline{v}{t-1}{alg:increase-energy-locked}{\timer} = 0$.
    In the previous paragraph, we proved that this execution implies there are no neighbors $u \in N_{t-1}(v)$ with $u_{t-1}.\maxttl = v_t.\maxttl = \varline{v}{t-1}{alg:increase-energy-locked}{\maxttl}$; thus, $S = \emptyset$ and $\varline{v}{t-1}{alg:timer-sync}{\timer} = \varline{v}{t-1}{alg:increase-energy-locked}{\timer} = 0$.
    But then $\varline{v}{t-1}{alg:timer-increment}{\timer} = \big(\varline{v}{t-1}{alg:timer-sync}{\timer} + 1\big) \bmod \varline{v}{t-1}{alg:increase-energy-locked}{\maxttl} = 1$, implying that $v.\source$ will not be updated again, yielding $v_t.\source = \varline{v}{t-1}{alg:increase-energy-locked}{\source} = \true$, a contradiction.
    So $v$ does not execute Line~\ref{alg:increase-energy-locked} in round $t-1$ and thus must execute Line~\ref{alg:increase-energy-recv}, yielding $\varline{v}{t-1}{alg:update-energy-cond}{\ttl} = \varline{v}{t-1}{alg:increase-energy-recv}{\ttl} = 0$.
    As we observed before, $\recvttl \geq 1$ and thus including $\varline{v}{t-1}{alg:update-energy-cond}{\ttl} = 0$ does not change $\recvttl$.
    Bringing this all together with $\recvttl \geq 1$, we conclude
    \begin{equation}
        v_t.\ttl = \varline{v}{t-1}{alg:update-energy-cond}{\ttl} = \recvttl - 1 = \underset{u \in M_{t-1}[v]}{\max} u_{t-1}.\ttl - 1 = \max\left\{0, \underset{u \in M_{t-1}[v]}{\max} u_{t-1}.\ttl - 1\right\}.
    \end{equation}

    Therefore, in all cases, the lemma holds.
\end{proof}

\lemttlupdating*
\begin{proof}
    Consider the execution of Algorithm~\ref{alg:deterministic-coloring} by any node $v \in V$ in any round $t \geq 0$ satisfying the conditions of the lemma.
    By supposition, $v$ does not update $v.\maxttl$, so it must not execute Lines~\ref{alg:increase-energy-recv} or~\ref{alg:increase-energy-locked}.
    Similarly, $v$ must not update $v.\ttl$ at Lines~\ref{alg:locked-bcast-energy} or~\ref{alg:become-set-energy}, as this implies $v_{t+1}.\source = \true$, a contradiction.
    So $v$ may only modify $v.\ttl$ at Line~\ref{alg:update-energy}.

    If $v$ does not execute Line~\ref{alg:update-energy}, then $v$ does not update $v.\ttl$ in round $t$ and thus $v_{t+1}.\ttl = v_t.\ttl$.
    This occurs only if $\varline{v}{t}{alg:timer-sync}{\state} = \reset$; it cannot be because $\varline{v}{t}{alg:timer-sync}{\state} \neq \reset$ but $\varline{v}{t}{alg:locked-signal-false}{\source} = \true$, as this would again contradict $v_{t+1}.\source = \false$.
    Since Line~\ref{alg:increase-energy-recv} is not run, $v_t.\state = \varline{v}{t}{alg:timer-sync}{\state} = \reset$.
    By Lemma~\ref{lem:zero-energy}, we have $v_t.\ttl = v_{t+1}.\ttl = 0$.
    Otherwise, $v$ executes Line~\ref{alg:update-energy}.
    This yields $v_{t+1}.\ttl = \varline{v}{t}{alg:update-energy}{\ttl} = \recvttl - 1 \leq e^*(t) - 1$, where the final inequality follows from the definition of $\recvttl$ (Line~\ref{alg:choose-max-ttl}).
    Combining these cases yields $e^*(t+1) \leq \max\{0, e^*(t) - 1\}$.

    To obtain a matching lower bound, first suppose $e^*(t) > 0$ and consider a node $v \in V$ with $v_t.\ttl = e^*(t)$; such a node must exist by definition of $e^*(t)$.
    We have already shown that $v_{t+1}.\ttl \leq e^*(t+1) \leq \max\{0, e^*(t) - 1\} < e^*(t) = v_t.\ttl$, so $v$ must execute Line~\ref{alg:update-energy} in round $t$.
    When it does so, combining the definitions of $e^*(t)$ and $\recvttl$ with the fact that $v_t.\ttl = \varline{v}{t}{alg:locked-signal-false}{\ttl} = e^*(t)$ is a candidate for $\recvttl$ implies $e^*(t) \geq \recvttl \geq v_t.\ttl = e^*(t)$, yielding $v_{t+1}.\ttl = \varline{v}{t}{alg:update-energy}{\ttl} = \recvttl - 1 = e^*(t) - 1$.
    Thus, if $e^*(t) > 0$, we have $e^*(t+1) \geq e^*(t) - 1 = \max\{0, e^*(t) - 1\}$.

    Finally, suppose $e^*(t) = 0$.
    Then by our upper bound, $e^*(t+1) \leq \max\{0, e^*(t) - 1\} = 0$.
    But Lemma~\ref{lem:ttl-bound} implies $e^*(t+1) \geq 0$, so $e^*(t+1) = 0 = \max\{0, e^*(t) - 1\}$.
\end{proof}

\lempersistnosignal*
\begin{proof}
    Let $t_1 = \min\{t \geq 0 : m^*(t) = m\}$ be the earliest time at which any node obtains $\maxttl = m$; note that this implies $t_1 \leq t_0$.
    We prove this lemma via a series of claims.

    \begin{claim} \label{claim:no-signal-fixed-m} 
        For all times $t \in [t_1, t_0 + 3m]$, $m^*(t) = m$.
    \end{claim}
    \begin{claimproof}
        Consider any time $t \in [t_1, t_0 + 3m]$.
        By the definition of $t_1$, the supposition that $m^*(t_0 + 3m) = m$, and the fact that \maxttl\ is nondecreasing by Lemma~\ref{lem:maxttl-nondecreasing}, we must have $m = m^*(t_1) \leq m^*(t) \leq m^*(t_0 + 3m) = m$.
    \end{claimproof}

    \begin{claim} \label{claim:no-signal-maxttl}
        For all nodes $v \in V$ and times $t \in [t_1 + m - 1, t_0 + 3m]$, $v_t.\maxttl = m$.
    \end{claim}
    \begin{claimproof}
        This is the same argument as in Claim~\ref{claim:all-maxttl}; see that claim for a full proof.
    \end{claimproof}

    \begin{claim} \label{claim:no-source}
        For all nodes $v \in V$ and times $t \in [t_0 + 2m - 1, t_0 + 3m]$, $v_t.\source = \false$.
    \end{claim}
    \begin{claimproof}
        Consider any node $v \in V$.
        We will first show that there is a time $t' \in [t_0 + m - 1, t_0 + 2m - 1]$ at which $v_{t'}.\source = \false$.
        This is trivially satisfied if $v_{t_0 + m - 1}.\source = \false$, so suppose $v_{t_0 + m - 1}.\source = \true$ and consider the execution of Algorithm~\ref{alg:deterministic-coloring} by node $v$ in round $t_0 + m - 1$.
        By Claim~\ref{claim:no-signal-maxttl}, $v_{t_0 + m - 1}.\maxttl = m$.
        This gives two insights: first, $v$ must not execute Lines~\ref{alg:increase-energy-recv} or~\ref{alg:increase-energy-locked} since this would contradict Claim~\ref{claim:no-signal-fixed-m}; second, we must have $v_{t_0 + m - 1}.\ttl = m > 0$ by Lemmas~\ref{lem:maxttl-nondecreasing} and~\ref{lem:source-ttl}, which further implies $v_{t_0 + m - 1}.\state \neq \reset$ by Lemma~\ref{lem:zero-energy}.
        So $v$ must increment $v.\timer$ in Line~\ref{alg:timer-increment}, yielding either $\varline{v}{t_0 + m - 1}{alg:timer-increment}{\timer} = \varline{v}{t_0 + m - 1}{alg:timer-sync}{\timer} + 1 > v_{t_0 + m - 1}.\timer$ or $\varline{v}{t_0 + m - 1}{alg:timer-increment}{\timer} = 0$.
        In the former case, tracing Algorithm~\ref{alg:deterministic-coloring} shows $v$ will make no other updates to any of its variables in this round, yielding $v_{t_0 + m}.\source = \true$, $v_{t_0 + m}.\maxttl = m$, and $v_{t_0 + m}.\timer > v_{t_0 + m - 1}.\timer$.
        But then this argument can be applied to $v$ again to further increase $v.\timer$ in the next round and likewise in any consecutive future round reached from this former case.
        By Lemma~\ref{lem:timer-lt-maxttl}, $v.\timer$ cannot exceed $m - 1$; thus, this argument can only be repeatedly applied at most $m - 1$ times before reaching a round $t' - 1 \leq t_0 + m - 1 + (m - 1) = t_0 + 2m - 2$ in which $\varline{v}{t'-1}{alg:timer-increment}{\timer} = 0$.
        This satisfies the condition on Line~\ref{alg:timer-zero-cond}, and since $t' - 1 \geq t_0$, $v$ will not witness a signal to satisfy the condition on Line~\ref{alg:locked-check-signal}.
        This yields $\varline{v}{t'-1}{alg:locked-signal-false}{\source} = \false$, which is not updated again in round $t' - 1$, yielding $v_{t'}.\source = \false$ at time $t' \leq t_0 + 2m - 1$.

        We conclude by arguing that $v_t.\source$ remains \false\ for all times $t \in [t', t_0 + 3m]$, where $t' \in [t_0 + m - 1, t_0 + 2m - 1]$ is the time identified above at which $v_{t'}.\source = \false$.
        Again, Claim~\ref{claim:no-signal-maxttl} implies $v_t.\maxttl = m$, ensuring $v$ does not execute Lines~\ref{alg:increase-energy-recv} or~\ref{alg:increase-energy-locked} to avoid contradicting Claim~\ref{claim:no-signal-fixed-m}.
        Also, $t \geq t_0$, so $v$ will not witness a signal to satisfy the conditions on Lines~\ref{alg:locked-check-signal} or~\ref{alg:become-set-cond}.
        Thus, $v$ cannot execute any lines that would make $v.\source = \true$, so $v_{t+1}.\source = v_t.\source = \false$, as desired.
        This argument can be repeated in any round $t \geq t'$ where Claims~\ref{claim:no-signal-fixed-m} and~\ref{claim:no-signal-maxttl} apply which are, at minimum, $[t_0 + 2m - 1, t_0 + 3m]$.
    \end{claimproof}

    \begin{claim} \label{claim:all-reset}
        For all nodes $v \in V$ and times $t \geq t_0 + 3m$, $v_t.\state = \reset$.
    \end{claim}
    \begin{claimproof}
        By Claims~\ref{claim:no-signal-fixed-m}--\ref{claim:no-source}, Lemma~\ref{lem:ttl-updating} applies to all times $t' \in [t_0 + 2m - 1, t_0 + 3m - 1]$.
        Successive applications of this lemma yields $e^*(t_0 + 2m + i) = \max\{0, e^*(t_0 + 2m - 1) - (i + 1)\}$ for all $i \in [0, m]$.
        By Claim~\ref{claim:no-signal-fixed-m}, Lemma~\ref{lem:ttl-bound}, and the definition of $e^*(\cdot)$, we have $e^*(t') \leq m^*(t') = m$ for all $t' \in [t_1, t_0 + 3m]$.
        Thus,
        \begin{equation}
            e^*(t_0 + 3m) = \max\{0, e^*(t_0 + 2m - 1) - m - 1\} \leq \max\{0, m - m - 1\} = 0.
        \end{equation}
        Combined with Lemma~\ref{lem:ttl-bound}, this implies $0 \leq v_{t_0 + 3m}.\ttl \leq e^*(t_0 + 3m) = 0$ for all nodes $v \in V$.
        Then, by Lemma~\ref{lem:zero-energy}, $v_{t_0 + 3m}.\state = \reset$ for all nodes $v \in V$.

        Now argue by induction on time $t \geq t_0 + 3m$ that $v_t.\state = \reset$ and $v_t.\maxttl = m$ for all nodes $v \in V$.
        The base case of $t = t_0 + 3m$ follows from the argument in the previous paragraph and Claim~\ref{claim:no-signal-maxttl}.
        So suppose the induction holds through time $t \geq t_0 + 3m$ and consider the execution of Algorithm~\ref{alg:deterministic-coloring} by any node $v$ in round $t$.
        By induction, all nodes have $\maxttl = m$ at time $t$, so $v$ does not run Line~\ref{alg:increase-energy-recv}.
        Also by induction, all nodes are \reset\ at time $t$, so $v$ does not run Lines~\ref{alg:increase-energy-locked} or~\ref{alg:join-set-nodes}.
        This yields $\varline{v}{t}{alg:locked-bcast-energy}{\state} = \varline{v}{t}{alg:timer-sync}{\state} = v_t.\state = \reset$ and $\varline{v}{t}{alg:locked-bcast-energy}{\maxttl} = v_t.\maxttl = m$.
        But $t \geq t_0$, so $v$ does not witness a signal and the condition on Line~\ref{alg:become-set-cond} is not satisfied.
        Thus, no other updates are performed and $v_{t+1}.\state = \reset$ and $v_{t+1}.\maxttl = m$, as desired.
    \end{claimproof}

    \begin{claim} \label{claim:reset-color}
        For all nodes $v \in V$ and times $t \geq 0$, if $v_t.\state = \reset$, then $v_t.\algcolor = c_\bot$.
    \end{claim}
    \begin{claimproof}
        Observe that every update to node's \algcolor\ (Lines~\ref{alg:increase-energy-recv}, \ref{alg:locked-new-color}, \ref{alg:no-energy-color}, \ref{alg:update-color}, and \ref{alg:become-set-color}) is accompanied by a corresponding update to its \signal\ (Lines~\ref{alg:increase-energy-recv}, \ref{alg:locked-new-signal}, \ref{alg:no-energy-signal}, \ref{alg:update-realsignal}, and~\ref{alg:become-set-realsignal}, respectively).
        So it suffices to argue by induction on time $t \geq 0$ that for any node $v \in V$, if $v_t.\state = \reset$, then $v_t.\signal = \bot$.
        At time $t = 0$, initialization sets $v_0.\state = \set \neq \reset$ for all nodes $v \in V$, so the claim holds vacuously (see Table~\ref{tab:algvars}).

        Now suppose the claim holds for any $t \geq 0$ and consider the execution of Algorithm~\ref{alg:deterministic-coloring} by any node $v$ in round $t$ resulting in $v_{t+1}.\state = \reset$.
        If $v$ executes Line~\ref{alg:no-energy-state} which yields $\varline{v}{t}{alg:locked-bcast-energy}{\state} = \reset$, then it also executes Line~\ref{alg:no-energy-signal} which yields $\varline{v}{t}{alg:locked-bcast-energy}{\signal} = \bot$.
        Since $v$ must remain \reset\ to satisfy the supposition that $v_{t+1}.\state = \reset$, it does not make any further updates to $v.\signal$ or $v.\state$, yielding $v_{t+1}.\signal = \bot$, as desired.

        Otherwise, if $v$ does not execute Line~\ref{alg:no-energy-state} but does execute Line~\ref{alg:increase-energy-recv} in round $t$, then we have $\varline{v}{t}{alg:increase-energy-recv}{\signal} = \bot$ and $\varline{v}{t}{alg:increase-energy-recv}{\state} = \reset$.
        It cannot execute Lines~\ref{alg:increase-energy-locked} or~\ref{alg:join-set-nodes}, as this would make it impossible to reach $v_{t+1}.\state = \reset$ without executing Line~\ref{alg:no-energy-state}, contradicting our supposition.
        Thus, $\varline{v}{t}{alg:timer-sync}{\state} = \reset$.
        This does not satisfy the condition on Line~\ref{alg:timer-increment-cond}, so $v$ does not update $v.\signal$ on Lines~\ref{alg:locked-new-signal}, \ref{alg:no-energy-signal}, or~\ref{alg:update-realsignal}.
        This yields $\varline{v}{t}{alg:locked-bcast-energy}{\signal} = \bot$ and $\varline{v}{t}{alg:locked-bcast-energy}{\state} = \reset$, which extends to $v_{t+1}.\signal = \bot$ as in the previous paragraph.

        The only remaining way to reach $v_{t+1}.\state = \reset$ is if $v_t.\state = \reset$ and $v$ does not update $v.\state$ in round $t$.
        By induction, $v_t.\signal = \bot$.
        Since $v$ does not update its \state\ in round $t$, it does not execute Line~\ref{alg:increase-energy-recv}, yielding $\varline{v}{t}{alg:increase-energy-recv}{\signal} = v_t.\signal = \bot$ and $\varline{v}{t}{alg:increase-energy-recv}{\state} = v_t.\state = \reset$.
        The remainder of round $t$ then progresses as in the previous paragraph, yielding $v_{t+1}.\signal = \bot$, as desired.
    \end{claimproof}

    Claims~\ref{claim:all-reset} and~\ref{claim:reset-color} together imply that $v_t.\algcolor = c_\bot$ for all nodes $v \in V$ and times $t' \geq t_0 + 3m$.
\end{proof}

\lempersistsetcolor*
\begin{proof}
    We prove this lemma via a series of claims.
    Our first two claims characterize the maximum \maxttl\ value and the type of nodes that hold it w.r.t.\ $t_1 = \min\{t \geq 0 : m^*(t) = m\}$, the earliest time at which any node obtains $\maxttl = m$.
    Note that this implies $t_1 \leq t_0$.

    \begin{claim} \label{claim:fixed-max-ttl}
        For all times $t \geq t_1$, $m^*(t) = m$.
    \end{claim}
    \begin{claimproof}
        We have $m^*(t) = m$ for all $t \geq t_0$ by supposition, so consider any $t \in [t_1, t_0)$.
        By the definition of $t_1$ and the fact that \maxttl\ is nondecreasing (Lemma~\ref{lem:maxttl-nondecreasing}), we must have $m = m^*(t_1) \leq m^*(t) \leq m^*(t_0) = m$.
    \end{claimproof}

    \begin{claim} \label{claim:first-max-ttl}
        For all nodes $v \in V$, if $v_{t_1}.\maxttl = m$, then $v_{t_1}.\state = \set$, $v_{t_1}.\source = \true$, and $v_{t_1}.\timer = 1$.
    \end{claim}
    \begin{claimproof}
        If $m = 2$, then initialization implies $t_1 = 0$ and $v_0.\state = \set$, $v_0.\source = \true$, $v_0.\maxttl = m$, and $v_0.\timer = 1$ for all nodes $v \in V$, as desired (see Table~\ref{tab:algvars}).
        Otherwise, $m > 2$ implies $t_1 > 0$.
        By definition, $t_1$ is the earliest time at which any node $v$ obtains $v.\maxttl = m$, so this must occur because such a node $v$ executes Line~\ref{alg:increase-energy-locked} in round $t_1 - 1$.
        This yields $\varline{v}{t_1 - 1}{alg:increase-energy-locked}{\state} = \set$, $\varline{v}{t_1 - 1}{alg:increase-energy-locked}{\source} = \true$, and $\varline{v}{t_1 - 1}{alg:increase-energy-locked}{\timer} = 0$.
        This also implies no neighbor $u$ could have sent $v$ a message with $u_{t_1-1}.\maxttl = \varline{v}{t_1 - 1}{alg:def-set-msgs}{\maxttl} = m$, so $v.\timer$ does not change in Line~\ref{alg:timer-sync}. 
        The only other update to $v.\timer$ is Line~\ref{alg:timer-increment}:
        \begin{equation}
            v_{t_1}^{(\ref{alg:timer-increment})}.\timer = (v_{t_1}^{(\ref{alg:increase-energy-locked})}.\timer + 1) \bmod v_{t_1}^{(\ref{alg:increase-energy-locked})}.\maxttl = (0 + 1) \bmod m = 1,
        \end{equation}
        using the supposition that $m > 2$.
        Since $\varline{v}{t_1 - 1}{alg:timer-increment}{\timer} = 1$, the condition on Line~\ref{alg:timer-zero-cond} is not satisfied and thus $\varline{v}{t_1 - 1}{alg:locked-signal-false}{\source} = \varline{v}{t_1 - 1}{alg:increase-energy-locked}{\source} = \true$.
        This, in turn, means the condition on Line~\ref{alg:update-energy-cond} is not satisfied, so $v.\timer$ is not reset on Line~\ref{alg:no-energy-timer}.
        Finally, $\varline{v}{t_1 - 1}{alg:locked-bcast-energy}{\state} = \varline{v}{t_1 - 1}{alg:increase-energy-locked}{\state} = \set$, so $v$ does not execute Lines~\ref{alg:become-set-realsignal}--\ref{alg:become-set-timer}.
        Therefore, we have $v_{t_1}.\state = \set$, $v_{t_1}.\source = \true$, and $v_{t_1}.\timer = 1$, as desired.
    \end{claimproof}

    Let $L_t = \{v \in V : v_t.\state = \locked \text{ and } v_t.\maxttl = m\}$ be the set of nodes that are \locked\ and hold the maximum \maxttl\ value.
    This is either none of the nodes or all of them (Claim~\ref{claim:universal-lock}) and all such nodes have the same \timer\ values (Claim~\ref{claim:lock-timer-sync}).

    \begin{claim} \label{claim:universal-lock}
        For all times $t \geq t_1$, either $L_t = \emptyset$ or $L_t = V$.
    \end{claim}
    \begin{claimproof}
        Suppose to the contrary that both $L_t$ and $\overline{L_t} = V \setminus L_t$ are nonempty at time $t \geq t_1$.
        Initially, all nodes $v \in V$ have $v_0.\state = \set$ (see Table~\ref{tab:algvars}), so we must have $t > 0$.
        Since $G_t$ is connected, there are nodes $u \in L_t$ and $v \in \overline{L_t}$ such that $\{u, v\} \in E_t$.
        So $u$ will send $v$ a message in round $t$.
        In particular, this message will reflect $u_t.\state = \locked$ and $u_t.\maxttl = m$ because $u \in L_t$; moreover, no other node has a larger \maxttl\ value because $m^*(t) = m$.
        Thus, the largest \maxttl\ value that $v$ receives in round $t$ is exactly $m$.
        
        If $v \in \overline{L_t}$ because $v_t.\maxttl \neq m$, then in fact $v_t.\maxttl < m$.
        So the message from $u$ causes Line~\ref{alg:increase-energy-recv} to execute, yielding $\varline{v}{t}{alg:increase-energy-recv}{\state} = \reset$ and $\varline{v}{t}{alg:increase-energy-recv}{\maxttl} = m$.
        Otherwise, $v \in \overline{L_t}$ because $v_t.\maxttl = m$ but $v_t.\state \neq \locked$.
        In either case, Line~\ref{alg:increase-energy-locked} will execute since $u_t.\state = \locked$, $u_t.\maxttl = m = \varline{v}{t}{alg:increase-energy-recv}{\maxttl}$, and $\varline{v}{t}{alg:increase-energy-recv}{\state} \neq \locked$.
        This yields $\varline{v}{t}{alg:increase-energy-locked}{\maxttl} = 2m \leq v_{t+1}.\maxttl$ by Lemma~\ref{lem:maxttl-nondecreasing}, contradicting Claim~\ref{claim:fixed-max-ttl}.
    \end{claimproof}

    \begin{claim} \label{claim:lock-timer-sync}
        For all times $t \geq t_1$, if $u,v \in L_t$, then $u_t.\timer = v_t.\timer$.
    \end{claim}
    \begin{claimproof}
        Argue by induction on $t \geq t_1$.
        At time $t_1$, any node $w \in V$ with $w_{t_1}.\maxttl = m$ has $w_{t_1}.\state = \set \neq \locked$ by Claim~\ref{claim:first-max-ttl}, so $L_{t_1} = \emptyset$ and the claim holds vacuously.

        Now suppose the claim holds for any $t \geq t_1$ and consider any two nodes $u, v \in L_{t+1}$; if none exist, we are done.
        By Claim~\ref{claim:universal-lock}, either $L_t = \emptyset$ or $L_t = V$.
        If the former, we will show both $u$ and $v$ execute Line~\ref{alg:become-locked} in round $t$, implying $u_{t+1}.\timer = v_{t+1}.\timer = 0$ by the condition in Line~\ref{alg:timer-zero-cond}.
        If the latter, we have $u_t.\timer = v_t.\timer$ by the induction hypothesis and will show that these values remain synchronized in round $t$.

        Suppose that $L_t = \emptyset$.
        Then for each $w \in \{u, v\}$, $w_t.\state \neq \locked$ or $w_t.\maxttl < m$.
        If $w_t.\state \neq \locked$, then we know $w$ must execute Line~\ref{alg:become-locked} in round $t$ to become \locked\ and must not later become \reset\ by executing Line~\ref{alg:no-energy-state} to satisfy the supposition that $w \in L_{t+1}$.
        This is only possible if $\varline{w}{t}{alg:timer-increment}{\timer} = 0$, and this \timer\ value is not changed again.
        So $w_{t+1}.\timer = 0$.
        If $w_t.\maxttl < m$, then in order to obtain $w_{t+1}.\maxttl = m$ and satisfy the supposition that $w \in L_{t+1}$, $w$ must execute Line~\ref{alg:increase-energy-recv} or~\ref{alg:increase-energy-locked} in round $t$.
        Either way, $w_{t+1}^{(\ref{alg:increase-energy-locked})}.\state \neq \locked$.
        This then falls into the prior case where $w$ must execute Line~\ref{alg:become-locked} to become \locked, again yielding $w_{t+1}.\timer = 0$.
        Thus, if $L_t = \emptyset$, then $u_{t+1}.\timer = v_{t+1}.\timer = 0$.

        Now suppose $L_t = V$.
        Then $u, v \in L_t$, implying $u_t.\timer = v_t.\timer$ by induction.
        By definition of $L_t$, we have $u_t.\state = v_t.\state = \locked$ and $u_t.\maxttl = v_t.\maxttl = m$.
        Tracing the execution of Algorithm~\ref{alg:deterministic-coloring} in round $t$ from these values, the only update $u$ and $v$ make to their \timer\ values is the increment in Line~\ref{alg:timer-increment}.
        Like in the previous paragraph, the \timer\ update in Line~\ref{alg:no-energy-timer} is not executed because it implies $u$ and $v$ are \reset, violating the supposition that $u, v \in L_{t+1}$.
        Thus, $u_{t+1}.\timer = (u_t.\timer + 1) \bmod u_t.\maxttl = (v_t.\timer + 1) \bmod v_t.\maxttl = v_{t+1}.\timer$.
    \end{claimproof}

    \begin{claim} \label{claim:continuous-signal}
        For all times $t \geq t_0$, if there is a node $v \in L_t$ with $v_t.\source = \true$, then there is a node $w \in L_{t+1}$ with $w_{t+1}.\source = \true$.
    \end{claim}
    \begin{claimproof}
        Consider any node $v \in L_t$ with $v_t.\source = \true$; if none exist, we are done.
        Since $L_t \neq \emptyset$, we know $L_t = V$ by Claim~\ref{claim:universal-lock}; then, by Claim~\ref{claim:lock-timer-sync}, we have $u_t.\timer = v_t.\timer$ for all $u \in V$.
        Tracing Algorithm~\ref{alg:deterministic-coloring} from these values we have $\varline{u}{t}{alg:timer-increment}{\state} = \locked$ for all nodes $u \in V$ and either the condition on Line~\ref{alg:timer-zero-cond} is not satisfied by any nodes---implying $\varline{v}{t}{alg:locked-signal-false}{\source} = v_t.\source = \true$---or it is satisfied by all nodes, implying some node $w \in L_t$ has $\varline{w}{t}{alg:locked-signal-false}{\source} = \true$ because the signal is persistent.
        In either case, there is a node $w \in L_t$ for which $\varline{w}{t}{alg:locked-signal-false}{\state} = \locked$ and $\varline{w}{t}{alg:locked-signal-false}{\source} = \true$.
        Again tracing Algorithm~\ref{alg:deterministic-coloring} from these values, $w$ will not update $w.\state$ or $w.\source$ again this round.
        Moreover, since $w_t.\maxttl = m$ is nondecreasing by Lemma~\ref{lem:maxttl-nondecreasing} and $m = m^*(t+1)$ by supposition, we must have $w_{t+1}.\maxttl = m$.
        Taken together, we conclude $w \in L_{t+1}$ and $w_{t+1}.\source = \true$.
    \end{claimproof}

    \begin{claim} \label{claim:set-to-locked}
        For all times $t \geq t_1$, if there is a node $v \in V$ with $v_t.\state = \set$, $v_t.\source = \true$, and $v_t.\maxttl = m$, then there is a time $t' \in [t + 1, t + m - v_t.\timer]$ such that $v_{t'}.\state = \locked$, $v_{t'}.\maxttl = m$, and $v_{t'}.\timer = 0$.
    \end{claim}
    \begin{claimproof}
        Consider any time $t \geq t_1$ and any node $v$ with $v_t.\state = \set$, $v_t.\source = \true$, and $v_t.\maxttl = m$; if none exist, we are done.
        When $v$ executes Algorithm~\ref{alg:deterministic-coloring} in round $t$, it cannot execute Lines~\ref{alg:increase-energy-recv} or~\ref{alg:increase-energy-locked} since it has $v_t.\maxttl = m$ and doing so would contradict Claim~\ref{claim:fixed-max-ttl}.
        Thus, regardless of whether it executes Line~\ref{alg:timer-sync}, $\varline{v}{t}{alg:timer-sync}{\state} = v_t.\state = \set$ and $\varline{v}{t}{alg:timer-sync}{\timer} \geq v_t.\timer$.
        After the \timer\ increment in Line~\ref{alg:timer-increment}, either $\varline{v}{t}{alg:timer-increment}{\timer} = \varline{v}{t}{alg:timer-sync}{\timer} + 1 > v_t.\timer$ or $\varline{v}{t}{alg:timer-increment}{\timer} = 0$.
        In the former case, tracing Algorithm~\ref{alg:deterministic-coloring} shows $v$ will make no other updates to any of its variables in round $t$, yielding $v_{t+1}.\state = \set$, $v_{t+1}.\source = \true$, $v_{t+1}.\maxttl = m$, and $v_{t+1}.\timer > v_t.\timer$.
        But then this argument can be applied to $v$ again to further increase $v.\timer$ in the next round and likewise in any consecutive future round reached from this former case.
        By Lemma~\ref{lem:timer-lt-maxttl}, $v.\timer$ cannot exceed $m - 1$; thus, this argument can only be repeatedly applied at most $m - 1 - v_t.\timer$ times before reaching a round $t' - 1 \leq t + m - 1 - v_t.\timer$ in which $\varline{v}{t'-1}{alg:timer-increment}{\timer} = 0$.
        This satisfies the condition on Line~\ref{alg:timer-zero-cond}, yielding $\varline{v}{t'-1}{alg:become-locked}{\state} = \locked$.

        It remains to show that $v$ does not further update its \state\ or \timer\ in round $t' - 1$.
        Such updates occur if and only if $v$ becomes \reset\ on Line~\ref{alg:no-energy-state}, which in turn occurs if and only if $\recvttl = 1$ (defined on Line~\ref{alg:choose-max-ttl}).
        Since $v_{t'-1}.\source = \true$, Lemma~\ref{lem:source-ttl} implies $v_{t'-1}.\ttl = v_{t'-1}.\maxttl = m$.
        Tracing Algorithm~\ref{alg:deterministic-coloring} shows $v_{t'-1}.\ttl = \varline{v}{t'-1}{alg:update-energy-cond}{\ttl} = m$ is a candidate for $\recvttl$, so $\recvttl \geq m > 1$ (by initialization).
        Therefore, we conclude that $v_{t'}.\state = \locked$, $v_{t'}.\maxttl = m$ (by Lemma~\ref{lem:maxttl-nondecreasing} and Claim~\ref{claim:fixed-max-ttl}), and $v_{t'}.\timer = 0$, satisfying the claim at time $t' \in [t + 1, t + m - v_t.\timer]$.
    \end{claimproof}

    \begin{claim} \label{claim:locked-bcaster}
        For all times $t \geq t_0$, if there is a node $v \in V$ with $v_t.\state = \set$, $v_t.\source = \true$, and $v_t.\maxttl = m$, then there is a time $t' \in [t + 1, t + m]$ and a node $w \in V$ with $w_{t'}.\state = \locked$, $w_{t'}.\source = \true$, and $w_{t'}.\maxttl = m$.
    \end{claim}
    \begin{claimproof}
        By Claim~\ref{claim:set-to-locked} and Lemma~\ref{lem:timer-lt-maxttl}, there is a time $t' \in [t + 1, t + m]$ such that $v \in L_{t'}$ and $v_{t'}.\timer = 0$.     
        By Claims~\ref{claim:universal-lock} and~\ref{claim:lock-timer-sync}, we have $u_{t'}.\timer = 0$ for all $u \in L_{t'} = V$.
        To reach these values, it must be that $\varline{u}{t'-1}{alg:become-locked}{\state} = \locked$ and $\varline{u}{t'-1}{alg:timer-increment}{\timer} = 0$ for all $u \in V$.
        Since the signal is persistent at time $t' - 1 \geq t \geq t_0$, at least one node $w \in V$ witnesses it, yielding $\varline{w}{t'-1}{alg:locked-signal-false}{\source} = \true$.
        This node will not update any other variables in round $t'-1$, yielding $w_{t'}.\state = \locked$, $w_{t'}.\source = \true$, and $w_{t'}.\maxttl = m$, as desired.
    \end{claimproof}
        
    \begin{claim} \label{claim:set-bcaster}
        For all times $t \geq t_1$, if there is a node $v \in V$ with $v_t.\state = \set$ and $v_t.\maxttl = m$, then there is a node $u \in V$ with $u_t.\state = \set$, $u_t.\source = \true$, and $u_t.\maxttl = m$.
    \end{claim}
    \begin{claimproof}
        Argue by induction on time $t \geq t_1$.
        At time $t_1$, Claim~\ref{claim:first-max-ttl} shows that any node $v \in V$ with $v_{t_1}.\maxttl = m$ also has $v_{t_1}.\state = \set$ and $v_{t_1}.\source = \true$, as desired.
        Now suppose the claim holds for any $t \geq t_1$ and consider any node $v \in V$ with $v_{t+1}.\state = \set$ and $v_{t+1}.\maxttl = m$; if none exist, then we are done.
        
        First suppose that $v_t.\state = \set$ and $v_t.\maxttl = m$.
        By induction, there is a node $u \in V$ with $u_t.\state = \set$, $u_t.\source = \true$, and $u_t.\maxttl = m$.
        Consider the execution of Algorithm~\ref{alg:deterministic-coloring} by this node in round $t$.
        By Lemma~\ref{lem:maxttl-nondecreasing} and Claim~\ref{claim:fixed-max-ttl}, we know $u_{t+1}.\maxttl = u_t.\maxttl = m$.
        So $u$ does not execute Lines~\ref{alg:increase-energy-recv} or~\ref{alg:increase-energy-locked}, yielding $\varline{u}{t}{alg:timer-sync}{\state} = u_t.\state = \set$ and $\varline{u}{t}{alg:timer-sync}{\source} = u_t.\source = \true$.
        If $u$ does not execute Line~\ref{alg:become-locked}, it also will not execute Lines~\ref{alg:locked-signal-false} or~\ref{alg:no-energy-state} because $\varline{u}{t}{alg:timer-sync}{\source} = \true$, yielding $u_{t+1}.\state = \varline{u}{t}{alg:locked-bcast-energy}{\state} = \varline{u}{t}{alg:timer-sync}{\state} = \set$ and $u_{t+1}.\source = \varline{u}{t}{alg:locked-bcast-energy}{\source} = \varline{u}{t}{alg:timer-sync}{\source} = \true$, satisfying the claim.
        Otherwise, if $u$ does execute Line~\ref{alg:become-locked} and become \locked, then it must also execute Line~\ref{alg:no-energy-state} to become \reset\ instead; without this, we would have $u_{t+1}.\state = \locked$, which by Claim~\ref{claim:universal-lock} implies $u \in L_{t+1} = V$, contradicting the existence of $v$ with $v_{t+1}.\state = \set$.
        But this cannot happen.
        Since $u_t.\source = \true$, we have $u_t.\ttl = u_t.\maxttl = m$ by Lemma~\ref{lem:source-ttl}.
        Thus, $\varline{u}{t}{alg:update-energy}{\ttl} = \recvttl - 1 = m - 1 > 0$ by Lemma~\ref{lem:maxttl-nondecreasing}, so the condition on Line~\ref{alg:no-energy} is not satisfied and Line~\ref{alg:no-energy-state} is not run, a contradiction.

        Otherwise, suppose that $v_t.\state \neq \set$ or $v_t.\maxttl \neq m$.
        Then $v$ must execute at least one of Lines~\ref{alg:increase-energy-locked}, \ref{alg:join-set-nodes}, or~\ref{alg:become-set-state} to satisfy $v_{t+1}.\state = \set$ and at least one of Lines~\ref{alg:increase-energy-recv} or~\ref{alg:increase-energy-locked} to satisfy $v_{t+1}.\maxttl = m$.
        Working backwards, if $v$ executes Line~\ref{alg:become-set-state}, then $v_{t+1}.\source = \varline{v}{t}{alg:become-set-signal}{\source} = \true$, completing the induction.
        Otherwise, if $v$ executes Line~\ref{alg:join-set-nodes}, then it must have received a message from a node $u$ with $u_t.\state = \set$ and $u_t.\maxttl = m$, so the analysis of the previous paragraph completes the induction.
        Otherwise, if $v$ executes Line~\ref{alg:increase-energy-locked}, then $\varline{v}{t}{alg:increase-energy-locked}{\source} = \true$, and this will persist until the end of the round by the same argument as in the previous paragraph.
        Finally, if $v$ executes Line~\ref{alg:increase-energy-recv}, then $\varline{v}{t}{alg:increase-energy-recv}{\state} = \reset$, so $v$ must also execute one of the lines in the previous cases to satisfy $v_{t+1}.\state = \set$, completing the induction.
    \end{claimproof}
        
    \begin{claim} \label{claim:no-bcaster}
        For all times $t \geq t_0$, if $v_t.\source = \false$ and $v_t.\maxttl = m$ for all $v \in V$, then there is a time $t' \in [t, t+m]$ and a node $u \in V$ with $u_{t'}.\source = \true$ and $u_{t'}.\maxttl = m$.
    \end{claim}
    \begin{claimproof}
        By Lemma~\ref{lem:maxttl-nondecreasing} and Claim~\ref{claim:fixed-max-ttl}, all nodes $v \in V$ have $v_{t'}.\maxttl = m$ for all $t' \in [t, t+m]$.
        So suppose to the contrary that for all nodes $v \in V$ and times $t' \in [t, t+m]$, $v_{t'}.\source = \false$.
        By repeatedly applying Lemma~\ref{lem:ttl-updating} to times $\{t, \ldots, t + m - 1\}$, the largest \ttl\ value held by any node at time $t + m$ is $\max\{0, e^*(t) - m\}$ which---since $e^*(t) \leq m$ by Lemma~\ref{lem:ttl-bound}---is in fact equal to $0$.
        Lemma~\ref{lem:ttl-bound} also says all \ttl\ values are nonnegative.
        Thus, $v_{t+m}.\ttl = 0$ for all $v \in V$.
        This, in turn, implies $v_{t+m}.\state = \varline{v}{t+m-1}{alg:locked-bcast-energy}{\state} = \reset$ for all $v \in V$ by Lemma~\ref{lem:zero-energy}.
        But $m > 1$ by Lemma~\ref{lem:maxttl-nondecreasing}, so $t + m - 1 > t \geq t_0$; i.e., the signal was present during round $t + m - 1$.
        So at least one node $u$ witnessed it, satisfying the condition in Line~\ref{alg:become-set-cond} and executing Line~\ref{alg:become-set-signal} to obtain $u_{t+m}.\source = \true$, a contradiction.
    \end{claimproof}
        
    \begin{claim} \label{claim:all-maxttl}
        For all nodes $v \in V$ and times $t \geq t_1 + m - 1$, $v_t.\maxttl = m$.
    \end{claim}
    \begin{claimproof}
        By the definition of time $t_1$, there is a node $u \in V$ with $u_{t_1}.\maxttl = m$; by Claim~\ref{claim:first-max-ttl}, we also have $u_{t_1}.\state = \set$, $u_{t_1}.\source = \true$, and $u_{t_1}.\timer = 1$.
        Thus, by Claim~\ref{claim:set-to-locked}, there is a time $t' \in [t_1 + 1, t_1 + m - 1]$ when $u_{t'}.\state = \locked$, $u_{t'}.\maxttl = m$, and $u_{t'}.\timer = 0$.
        So $u \in L_{t'} = V$ by Claim~\ref{claim:universal-lock}, implying that $v_{t'}.\maxttl = m$ for all nodes $v \in V$.
        By Lemma~\ref{lem:maxttl-nondecreasing} and Claim~\ref{claim:fixed-max-ttl}, $v_t.\maxttl = m$ for all times $t \geq t_1 + m - 1 \geq t'$.
    \end{claimproof}
        
    \begin{claim} \label{claim:reset-maxttl-time}
        For all times $t \geq t_1$, if there is a node $v \in V$ with $v_t.\state = \reset$ and $v_t.\maxttl = m$, then $t \geq t_1 + m$.
    \end{claim}
    \begin{claimproof}
        We will show that for all times $t \geq t_1$, each node $u \in V$ with $u_t.\maxttl = m$ has $u_t.\ttl \geq m - (t - t_1)$.
        This implies the claim: if there is a node $v \in V$ with $v_t.\state = \reset$ and $v_t.\maxttl = m$, then $v_t.\ttl = 0$ by Lemma~\ref{lem:zero-energy}, so we conclude $t \geq t_1 + m$.

        Argue by induction on $t \geq t_1$.
        If $t = t_1$, then by Claim~\ref{claim:first-max-ttl}, any node $v \in V$ with $v_{t_1}.\maxttl = m$ also has $v_{t_1}.\source = \true$.
        By Lemma~\ref{lem:source-ttl}, this implies $v_{t_1}.\ttl = v_{t_1}.\maxttl = m = m - (t - t_1)$, as desired.

        Now suppose the claim holds for any $t \geq t_1$ and consider any node $v \in V$ with $v_{t+1}.\maxttl = m$.
        Then $v_{t+1}.\ttl$ is characterized by Lemma~\ref{lem:ttl-value}.
        Since $t + 1 > t_1 \geq 0$, we can ignore the $t = 0$ case.
        If $v_{t+1}.\source = \true$, then $v_{t+1}.\ttl = v_{t+1}.\maxttl = m > m - (t + 1 - t_1)$, as desired.
        Otherwise, by the induction hypothesis,
        \begin{equation}
            v_{t+1}.\ttl = \max\left\{0, \max_{u \in M_t[v]} u_t.\ttl - 1\right\}
            \geq m - (t - t_1) - 1
            \geq m - (t + 1 - t_1),
        \end{equation}
        as desired.
        This completes the induction and thus proves the claim.
    \end{claimproof}
        
    \begin{claim} \label{claim:all-locked}
        For all times $t \geq t_0 + 2m$, $L_t = V$.
    \end{claim}
    \begin{claimproof}
        By Claims~\ref{claim:universal-lock} and~\ref{claim:continuous-signal}, it suffices to identify a node $v \in V$ and a time $t' \in [t_0, t_0 + 2m]$ such that $v \in L_{t'}$ and $v_{t'}.\source = \true$.
        If such a node exists at time $t' = t_0$, we are done.

        Suppose instead there is a node $u \in V$ such that $u_{t_0}.\state = \set$ and $u_{t_0}.\maxttl = m$.
        By Claim~\ref{claim:set-bcaster}, there is a node $v \in V$ such that $v_{t_0}.\state = \set$, $v_{t_0}.\source = \true$, and $v_{t_0}.\maxttl = m$.
        Then by Claim~\ref{claim:locked-bcaster}, there is a node $w \in V$ and a time $t' \in [t_0 + 1, t_0 + m]$ such that $w \in L_{t'}$ and $w_{t'}.\source = \true$, as desired.

        Now suppose there is a node $u \in L_{t_0}$ (implying $L_{t_0} = V$ by Claim~\ref{claim:universal-lock}) but $u_{t_0}.\source = \false$ for all $u \in V$.
        By Claim~\ref{claim:no-bcaster}, there is a node $v \in V$ and a time $t' \in [t_0, t_0 + m]$ such that $v_{t'}.\source = \true$ and $v_{t'}.\maxttl = m$.
        By Lemmas~\ref{lem:maxttl-nondecreasing} and~\ref{lem:source-ttl}, $v_{t'}.\ttl = v_{t'}.\maxttl > 1$.
        This, in turn, implies $v_{t'}.\state \neq \reset$ by Lemma~\ref{lem:zero-energy}.
        If $v_{t'}.\state = \locked$, we are done; otherwise, $v_{t'}.\state = \set$.
        Again by Claim~\ref{claim:locked-bcaster}, there is a node $w \in V$ and a time $t'' \in [t' + 1, t' + m] \leq t_0 + 2m$ such that $w \in L_{t''}$ and $w_{t''}.\source = \true$, as desired.

        Finally, suppose that none of the previous cases apply.
        By Claim~\ref{claim:fixed-max-ttl}, there is some node $v \in V$ with $v_{t_0}.\maxttl = m$.
        Since none of the previous cases apply, $v_{t_0}.\state = \reset$.
        This implies $t_0 \geq t_1 + m$ by Claim~\ref{claim:reset-maxttl-time}, which in turn implies all nodes $u \in V$ have $u_{t_0}.\maxttl = m$ by Claim~\ref{claim:all-maxttl} and $u_{t_0}.\state = \reset$ by the supposition that none of the previous cases apply.
        Tracing Algorithm~\ref{alg:deterministic-coloring} from these values, all nodes $u \in V$ have $\varline{u}{t_0}{alg:locked-bcast-energy}{\state} = \reset$.
        Since the signal $s \neq \bot$ is persistent at time $t_0$, at least one node $w \in V$ witnesses it, satisfying the condition on Line~\ref{alg:become-set-cond} and yielding $w_{t_0 + 1}.\state = \set$ and $w_{t_0 + 1}.\source = \true$; also, $w_{t_0 + 1}.\maxttl = m$ by Lemma~\ref{lem:maxttl-nondecreasing} and Claim~\ref{claim:fixed-max-ttl}.
        Thus, by Claim~\ref{claim:locked-bcaster}, there is a node $x \in V$ and a time $t' \in [t_0 + 2, t_0 + 1 + m] < t_0 + 2m$ such that $x \in L_{t'}$ and $x_{t'}.\source = \true$, as desired.
    \end{claimproof}
        
    \begin{claim} \label{claim:source-color}
        For all times $t \geq t_0 + 3m$, if $v_t.\source = \true$ then $v_t.\algcolor = c_s$.
    \end{claim}
    \begin{claimproof}
        By Claim~\ref{claim:all-locked}, $v \in L_t$ for all nodes $v \in V$ and all times $t \geq t_0 + 2m$.
        Moreover, by Claim~\ref{claim:lock-timer-sync}, there is a fixed value $\timer_t \in \{0, \ldots, m - 1\}$ for each $t \geq t_0 + 2m$ such that $v_t.\timer = \timer_t$ for all nodes $v \in L_t = V$.
        Tracing Algorithm~\ref{alg:deterministic-coloring} from these values, it is clear that $\timer_t$ increments modulo $m$ once per round.
        Let $t' \in [t_0 + 2m, t_0 + 3m)$ be the earliest time at which $\timer_{t'} = m - 1$.
        Clearly, all nodes $v \in L_{t'} = V$ will set
        \begin{equation}
            \varline{v}{t'}{alg:timer-increment}{\timer} = \left(\varline{v}{t'}{alg:timer-sync}{\timer} + 1\right) \bmod \varline{v}{t'}{alg:timer-sync}{\maxttl} = (\timer_{t'} + 1) \bmod m = 0,
        \end{equation}
        satisfying the condition on Line~\ref{alg:timer-zero-cond}.
        All nodes $v \in V$ then set $\varline{v}{t'}{alg:become-locked}{\state} = \locked$ and either set $\varline{v}{t'}{alg:locked-new-color}{\algcolor} = c_s$ and $\varline{v}{t'}{alg:locked-signal-true}{\source} = \true$ if they witness the signal or $\varline{v}{t'}{alg:locked-signal-false}{\source} = \false$ otherwise.
        The witnesses $w \in W_{t'}$ will not change their \algcolor\ or \source\ values again in round $t'$ since they do not satisfy the conditions on Lines~\ref{alg:update-energy-cond} or~\ref{alg:become-set-cond}, yielding $w_{t'+1}.\algcolor = c_s$ and $w_{t'+1}.\source = \true$, as desired.
        Any non-witnesses $v \not\in W_t$ do not satisfy Line~\ref{alg:become-set-cond} either, so they end round $t'$ with $v_{t'+1}.\source = \false$.

        Thus, the claim holds at time $t' + 1 \in (t_0 + 2m, t_0 + 3m]$.
        Now suppose by induction that the claim holds for some $t \geq t_0 + 3m$.
        Consider any node $v$ with $v_{t+1}.\source = \true$; if none exist, we are done.
        Since $t > t_0 + 2m$, Claim~\ref{claim:all-locked} implies $v \in L_t$ and $v \in L_{t+1}$.
        The former implies that when executing Algorithm~\ref{alg:deterministic-coloring} in round $t$, $v$ does not make any variable updates through Line~\ref{alg:timer-sync}, yielding $v_t.\algcolor = \varline{v}{t}{alg:timer-sync}{\algcolor}$ and $v_t.\source = \varline{v}{t}{alg:timer-sync}{\source}$.
        The latter means that when tracing this same execution backwards, $v$ did not satisfy the condition on Line~\ref{alg:become-set-cond} and thus $v_{t+1}.\source = \varline{v}{t}{alg:locked-bcast-energy}{\source} = \true$ and $v_{t+1}.\algcolor = \varline{v}{t}{alg:locked-bcast-energy}{\algcolor}$.
        Moreover, since no updates to \source\ occur in Lines~\ref{alg:update-energy-cond}--\ref{alg:locked-bcast-energy}, we must have $\varline{v}{t}{alg:locked-signal-false}{\source} = \varline{v}{t}{alg:locked-bcast-energy}{\source} = \true$ and $\varline{v}{t}{alg:locked-signal-false}{\algcolor} = \varline{v}{t}{alg:locked-bcast-energy}{\algcolor}$. 
        This is only possible if either $v$ executes both Lines~\ref{alg:locked-new-color} and~\ref{alg:locked-signal-true}, yielding $v_{t+1}.\algcolor = \varline{v}{t}{alg:locked-signal-false}{\algcolor} = c_s$, as desired, or if $v$ does not in fact make any updates to $v.\algcolor$ or $v.\source$ at all in round $t$, yielding $v_{t+1}.\algcolor = \varline{v}{t}{alg:timer-sync}{\algcolor} = v_t.\algcolor = c_s$ by the induction hypothesis.
        Either case completes the induction.
    \end{claimproof}
        
    \begin{claim} \label{claim:non-source-color}
        For all times $t = t_0 + 3m + i$ where $i \geq 0$, if $v_t.\ttl \geq \max\{1, m - i\}$, then $v_t.\algcolor = c_s$.
    \end{claim}
    \begin{claimproof}
        Argue by induction on $i \geq 0$.
        If $i = 0$, then $t = t_0 + 3m$.
        Consider any node $v$ with $v_t.\ttl \geq \max\{1, m - i\} = m$; if none exist, we are done.
        By Claim~\ref{claim:all-locked}, $v_t.\maxttl = m$.
        By Lemma~\ref{lem:source-ttl}, this implies $v_t.\source = \true$, which in turn implies $v_t.\algcolor = c_s$ by Claim~\ref{claim:source-color}.

        Now suppose the claim holds for any $i \geq 0$ and let $t = t_0 + 3m + i$.
        Consider any node $v$ with $v_{t+1}.\ttl \geq \max\{1, m - (i + 1)\}$; again, if none exist, we are done.
        If $v_{t+1}.\ttl = m$, then we have $v_{t+1}.\algcolor = c_s$ by the same argument as in the base case.
        Otherwise, if $v_{t+1}.\ttl < m$, then $v_{t+1}.\ttl \neq v_{t+1}.\maxttl$ by Claim~\ref{claim:all-locked} and thus $v_{t+1}.\source = \false$ by Lemma~\ref{lem:source-ttl}.
        We also have $v_{t+1}.\ttl \geq 1$, so $v_{t+1}.\state \neq \reset$ by Lemma~\ref{lem:zero-energy}.
        Tracing the execution of Algorithm~\ref{alg:deterministic-coloring} by $v$ in round $t$ backwards from $v_{t+1}.\state \neq \reset$ and $v_{t+1}.\source = \false$ shows $\varline{v}{t}{alg:timer-sync}{\state} \neq \reset$ and $\varline{v}{t}{alg:locked-signal-false}{\source} = \false$, satisfying the conditions on Lines~\ref{alg:timer-increment-cond} and~\ref{alg:update-energy-cond}, respectively.
        So $v_{t+1}.\ttl = \varline{v}{t}{alg:update-energy}{\ttl} = \recvttl - 1$, where $\recvttl$ (Line~\ref{alg:choose-max-ttl}) is whichever is larger: $\varline{v}{t}{alg:update-energy-cond}{\ttl}$ or the largest \ttl\ value $v$ received from any of its \set\ or \locked\ neighbors $u$ with $u_t.\maxttl = m$ at the start of round $t$.
        Since $L_t = V$ by Claim~\ref{claim:all-locked}, this complicated formulation simplifies to $\recvttl = \max_{u \in V} u_t.\ttl$.
        Consider any node $u$ with $u_t.\ttl = \recvttl$; if there are multiple, choose any with the largest $u_t.\signal$.
        Then
        \begin{equation}
            u_t.\ttl = \recvttl = v_{t+1}.\ttl + 1 \geq \max\{1, m - (i + 1)\} + 1 \geq \max\{1, m - i\}.
        \end{equation}
        By the induction hypothesis, $u_t.\algcolor = c_s$.
        Since we supposed $v_{t+1}.\ttl < m$, $v$ must execute Lines~\ref{alg:update-realsignal} and~\ref{alg:update-color}, adopting the \signal\ and \algcolor\ of $u$.
        Thus, $v_{t+1}.\algcolor = c_s$.
    \end{claimproof}

    We are now ready to prove the full lemma.
    Consider any node $v \in V$ and time $t \geq t_0 + 4m$, which can be rewritten as $t = t_0 + 3m + i$ for some $i \geq m$.
    By Claim~\ref{claim:all-locked}, $v_t.\state = \locked$.
    By Lemmas~\ref{lem:zero-energy} and~\ref{lem:ttl-bound}, this implies $v_t.\ttl \geq 1$.
    Since $v_t.\ttl \geq 1 = \max\{1, m - i\}$, we conclude that $v_t.\algcolor = c_s$ by Claim~\ref{claim:non-source-color}.
\end{proof}

\lemmaxttlupperbound*
\begin{proof}
Suppose for contradiction there exists a time $t_0$ such that $m^*(t_0) > 2n$ and $t_0$ is the minimal such time.
This implies $m^*(t_0) > 2$ and thus $\maxttl$ values have increased after initialization.
The only way for this to occur is for a node $v$ with $\varline{v}{t_0 - 1}{alg:increase-energy-recv}{\maxttl} = \frac{1}{2}m^*(t_0)$ to run line~\ref{alg:increase-energy-locked} in round $t_0 - 1$.
Let $m = m^*(t_0 - 1)$ and $t_1 = \min\{t \geq 0 : m^*(t) = m\}$.
Note that $m < m^*(t_0)$ by minimality of $t_0$ and since $m^*$ can only be decreased by a doubling at line~\ref{alg:increase-energy-locked} it must be that $m = \frac{1}{2}m^*(t_0) > n$.

For convenience, define 
$$M_t = \{v \in V : v_t.\maxttl = m\}$$
$$W_t = \{v \in M_t \mid v_t.\source = \true\}$$
$$s^*(t) = \max_{v \in M_t \mid v_t.\state = \set} (v_t.\timer)$$
$$S_t = \{v \in M_t \mid v_t.\state = \set \text{ and } v_t.\timer = s^*(t) \}$$
and for any set $X$ let $X_t^{\leq i} = \{v \in X_t \mid v_t.\ttl \leq i\}$.
We prove that for all $t \geq t_1$ either
\begin{enumerate}
	\item
	\label{cond:locked}
	\begin{enumerate}
		\item
		\label{cond:all-maxttl}
		$M_t = V$,
		
		\item
		\label{cond:all-locked}
		$v_t.\state = \locked$ for all $v \in V$,
		
		\item
		\label{cond:timer-sync}
		$v_t.\timer = w_t.\timer$ for all $v,w \in V$, and exactly one of the following cases is true:
			
		\item
		\label{cond:locked-source} 
		If there exists $v \in V$ with $v_t.\source = \true$, then $\lvert V_t^{\leq i} \rvert \leq \max(0, i - m + n)$ for $i = 0, \dots m$
		
		\item
		\label{cond:locked-no-source}
		Otherwise, if $v_t.\source = \false$ for all $v \in V$, then 
		\begin{enumerate}
			\item
			\label{cond:small-ttl}
			$\lvert V_t^{\leq i} \rvert \leq \max(0, i -m + n)$ for $i = 0,\dots,m-v_t.\timer-2$,
			
			\item
			\label{cond:large-ttl}
			$v_t.\ttl \leq m - v_t.\timer - 1$ for all $v \in V$, and
		\end{enumerate}
	\end{enumerate}
	Or
	\item
	\label{cond:setting}
	\begin{enumerate}
		\item
		\label{cond:small-maxttl}
		$v_t.\maxttl \leq m$ for all $v \in V$.
	
		\item
		\label{cond:not-locked}
		$v_t.\state \neq \locked$ for all $v \in M_t$,
		
		\item
		\label{cond:set-source}
		There exists $v \in S_t$ such that $v_t.\source = \true$.
		
		\item
		\label{cond:num-set-ttls}
		$\lvert S_t^{\leq i} \rvert \leq \max(0, i-m+|S_t|)$ for $i = 0,\dots,m$, and
		
		\item
		\label{cond:num-set}
		$\lvert S_t \rvert \geq \min(n, s^*(t))$.
	\end{enumerate}
\end{enumerate}

We prove by induction.
As a base case consider time $t_1$ and we will show that condition~\ref{cond:setting} is true.
Condition~\ref{cond:small-maxttl} must be true since $m^*(t_1) = m$ by definition.
Now let $v \in M_{t_1}$.
By the same logic as claim~\ref{claim:first-max-ttl} in Lemma~\ref{lem:persist-set-color}, $v_{t_1}.\state = \set$, $v_{t_1}.\source = \true$, and $v_{t_1}.\timer = 1$.
By Lemma~\ref{lem:source-ttl}, this also implies that $v_{t_1}.\ttl = v_{t_1}.\maxttl = m$.
Since this is true for all $v \in M_{t_1}$, clearly \ref{cond:not-locked} and~\ref{cond:set-source} are true.
Note $|S_{t_1}^{\leq m}| = |S_t| = \max(0, m - m +|S_t|)$.
For any $i < m$, $|S_{t_1}^{\leq i}| = 0 \leq \max(0, i - m + |S_t|)$.
Thus condition~\ref{cond:num-set-ttls} is true.
Finally, $|S_{t_1}| \geq 1 = s^*(t_1) \geq \min(n, s^*(t_1))$ so \ref{cond:num-set} is true.

Suppose that we have proven up to some $t \geq 0$.
We will show that one of these two cases will still be true at time $t+1$.
First, suppose that condition~\ref{cond:locked} was true at time $t$.
Consider the execution of the algorithm by any node $v$ in round $t$.
By condition~\ref{cond:all-maxttl}, no node receives a $\maxttl$ value different than its own, so line~\ref{alg:increase-energy-recv} is not run.
By condition~\ref{cond:all-locked}, all nodes are locked, so the first line nodes will execute is line~\ref{alg:timer-increment}
For all $v, w \in V$, $\varline{v}{t}{alg:timer-increment-cond}{\timer} = \varline{w}{t}{alg:timer-increment-cond}{\timer}$ by condition~\ref{cond:timer-sync} and because nodes did not change  $\timer$ in the execution up to this point.
We split into several cases at this point.

\begin{itemize}
	\item
	\textit{Case 1}.
	Condition~\ref{cond:locked-no-source} was true at time $t$ and all nodes have $\varline{v}{t}{alg:timer-zero-cond}{\timer} \leq m-2$.
	
	We will show that condition~\ref{cond:locked} with subcondition~\ref{cond:locked-no-source} will still be true at time $t+1$.
	Clearly the condition on line~\ref{alg:timer-zero-cond} is $\false$ for all nodes.
	Since all nodes have $v_t.\source = \false$, and $\source$ values have not changed, they all run line~\ref{alg:update-energy}.
	We claim that $\varline{v}{t}{alg:update-energy}{\ttl} > 0$ for all nodes $v$.
	If this fails for some $v$, then by definition of $l^*$ and since all nodes have non-zero \ttl, $v$ and all its neighbors had $v_t.\ttl = 1$.
	By the assumption of our case, it must be that $v_t.\timer = m - 3$ for all nodes.
	Thus $m - v_t.\timer - 2 = 1$ and so by induction, condition~\ref{cond:small-ttl}, and since $m \geq n + 1$, $|V_t^{\leq 1}| \leq \max(0, 1 - m + n) = 0$.
	So thus it is impossible for a node to have $\varline{v}{t}{alg:update-energy}{\ttl} = 0$.
	Thus nodes do not change any variables for the rest of the algorithm except maybe $\signal$ and $\algcolor$.
	Thus, conditions~\ref{cond:all-maxttl}, \ref{cond:all-locked}, and \ref{cond:timer-sync} will all be true at time $t+1$.
	Also, all nodes will have $v_{t+1}.\source = \false$ so it remains to show that condition~\ref{cond:locked-no-source} is true at $t+1$.
	
	Let $r$ be the unique (by condition~\ref{cond:timer-sync}) $\timer$ value for nodes at time $t$ and thus $r+1$ will be the \timer\ value at time $t+1$.
	First, we will show that condition~\ref{cond:small-ttl} holds at time $t+1$.
	Let $i \in [0, m - (r+1) - 2]$ and consider a node $v \in V_{t+1}^{\leq i}$.
	Since~\ref{alg:update-energy} is the last place nodes update $\ttl$ in round $t$, $v$ must have had $l^* \leq i + 1$ in round $t$.
	Since $v_t.\ttl$ is a candidate for $l^*$, this implies $v_t.\ttl \leq l^* \leq i + 1$.
	Thus $v \in V_t^{\leq i + 1}$ so $V_{t+1}^{\leq i} \subseteq V_t^{\leq i + 1}$.
	If $|V_t^{\leq i + 1}| = 0$, this implies that $|V_{t+1}^{\leq i}| = 0 \leq \max(0, i - m + n)$ and we are done.
	Otherwise, $|V_t^{\leq i+1}| \geq 1$.
	Note that $i+1 \leq m - r - 2$ and thus by induction and condition~\ref{cond:small-ttl}, $|V_t^{\leq i+1}| \leq \max(0, (i+1)-m+n)$.
	Thus $(i+1)-m+n \geq 1$; we will use this fact later.
	Also, since $i \leq m - r - 3$ and $r \geq 0$, we have $(i + 1) - m + n \leq n - 2$.
	Thus, $V_t^{\leq i + 1} \neq V$.
	Since $V_t^{\leq i + 1}$ and its complement are nonempty and $G$ is connected at time $t$, there must be nodes $v \in V_t^{\leq i+1}$ and $w \notin V_t^{\leq i+1}$ such that $vw \in E_t$.
	Thus, during round $t$, $v$ must have $l^* \geq w_t.\ttl > i + 1$.
	Since line~\ref{alg:update-energy} is the last place $v$ updates $\ttl$ in round $t$, we have $v_{t+1}.\ttl \geq l^* - 1  > i$, and thus $v \notin V_{t+1}^{\leq i}$.
	We now know that $V_{t+1}^{\leq i} \subseteq V_t^{\leq i+1} \setminus \{v\}$.
	Thus, 
	$$|V_{t+1}^{\leq i}| \leq |V_t^{\leq i +1}| - 1 \leq \max(0, (i+1)-m+n) - 1 = \max(0, i - m + n)$$
	where this final equality is because $(i+1)-m+n \geq 1$ as shown earlier.
	
	Next, we show that condition~\ref{cond:large-ttl} holds at time $t+1$.
	By induction, all nodes have $v_t.\ttl \leq m - r - 1$.
	Thus any node $w$ will have $l^* \leq m - r - 1$ in round $t$ and so $w_{t+1}.\ttl = l^* - 1 \leq m - r - 2$.
	
	\item
	\textit{Case 2}.
	Condition~\ref{cond:locked-no-source} was true at time $t$ and all nodes have $\varline{v}{t}{alg:timer-zero-cond}{\timer} = m-1$.
	
	We will show that condition~\ref{cond:setting} will be true at time $t+1$.
	As in case 1, the condition on line~\ref{alg:timer-zero-cond} will be $\false$ for all nodes.
	Thus all nodes run line~\ref{alg:update-energy}.
	Since $v_t.\timer +1 \mod{m} = m - 1$ and $m \geq 2$ (by Lemma~\ref{lem:maxttl-nondecreasing}), $v_t.\timer = m - 2$.
	Thus, by induction, condition~\ref{cond:small-ttl} says $|V_t^{\leq 0}| \leq \max(0, 0 - m + n) = 0$ since $n < m$ and condition~\ref{cond:large-ttl} says all nodes have $v_t.\ttl \leq 1$.
	Thus, all nodes $v$ must have $v_t.\ttl = 1$.
	During the execution in round $t$, all nodes therefore set $l^* = 1$ and thus $\varline{v}{t}{alg:update-energy}{\ttl} = l^* - 1 = 0$.
	Thus, all nodes have $\varline{v}{t}{alg:no-energy-state}{\state} = \reset$ and $\varline{v}{t}{alg:no-energy-timer}{\timer} = 0$.
	Since nodes can not become $\locked$ after these lines, $v_{t+1} \neq \locked$ for all $v$, and thus condition~\ref{cond:not-locked} is true.
	Note that all nodes end with $v_{t+1}.\timer = 0$ and so $s^*(t+1) = 0$.
	Any node with $v \in S_{t+1}$ must have run line~\ref{alg:become-set-energy}, and thus has $v_{t+1}.\ttl = v_{t+1}.\maxttl = m$.
	Thus $|S_t^{\leq i}| = 0 \leq \max(0, i - m + |S_t|)$ for any $i \neq m$, and $|S_t^{\leq m}| \leq \max(0, m - m + |S_t|) = |S_t|$ clearly must be true.
	Thus condition~\ref{cond:num-set-ttls} is true.
	Finally, $\min(n, s^*(t)) = 0$ so condition~\ref{cond:num-set} is true and all nodes have $v_{t+1}.\maxttl = m$ so condition~\ref{cond:small-maxttl} is true.
	
	\item
	\textit{Case 3}.
	Condition~\ref{cond:locked-source} was true at time $t$, $\varline{v}{t}{alg:timer-zero-cond}{\timer} = 0$ for all $v \in V$, and there is no signal in round $t$.
	
	We will prove that condition~\ref{cond:locked} with subcondition~\ref{cond:locked-no-source} will be true at time $t+1$.
	Clearly, by the conditions of this case, all nodes have $\varline{v}{t}{alg:locked-signal-false}{\source} = \false$.
	Thus, all nodes will execute line~\ref{alg:update-energy}.
	Proceeding in a very similar way to case 1, we first note that all nodes have $\varline{v}{t}{alg:update-energy}{\ttl} > 0$ since $|V_t^{\leq 1}| \leq \max(0, 1 - m + 1) = 0$.
	Then conditions~\ref{cond:all-maxttl}, \ref{cond:all-locked}, and \ref{cond:timer-sync} will all be true at time $t+1$.
	Also, by the assumptions of the case, all nodes have $v_{t+1}.\timer = 0$.
	For $i \in [0, m - 2]$, $V_{t+1}^{\leq i} \subseteq V_t^{\leq i + 1}$.
	If $V_t^{\leq i + 1} = \emptyset$, then $V_{t+1}^{\leq i} = 0 \leq \max(0, i - m + n)$.
	Otherwise, $1 \leq |V_t^{\leq i+1}| \leq \max(0, (i+1)-m+n) \leq n - 1$ where the second inequality is by induction and thus $i+1-m+n \geq 1$.
	Thus, there are nodes $v \in V_t^{\leq i+1}$ and $w$ with $w \notin V_t^{\leq i + 1}$ such that $vw \in E_t$.
	This implies that $v_{t+1}.\ttl \geq w_t.\ttl - 1 \geq i + 1$.
	Thus $V_{t+1}^{\leq i} \subseteq V_t^{\leq i+1} \setminus \{v\}$ and so $|V_{t+1}^{\leq i}| \leq |V_t^{\leq i+1}| - 1 \leq \max(0, i+1-m+n) - 1 = \max(0, i-m+n)$.
	Since this is true for all $i = 0, \ldots, m - 2$ and $m - v_{t+1}.\timer - 2 = m - 2$ for all $v$, condition~\ref{cond:small-ttl} will be true at $t+1$.
	Note that all nodes have $v_t.\ttl \leq v_t.\maxttl = m$ by Lemma~\ref{lem:ttl-bound}.
	Thus, all nodes have $v_{t+1}.\ttl = \varline{v}{t}{alg:update-energy}{\ttl} = l^* - 1 \leq m - 1 = m - v_{t+1}.\timer - 1$.
	Thus, condition~\ref{cond:large-ttl} is also true at time.
	
	\item
	\textit{Case 4}.
	Condition~\ref{cond:locked-source} was true at time $t$, and $\varline{v}{t}{alg:timer-zero-cond}{\timer} \neq 0$ for all $v \in V$, or there is a signal in round $t$.
	
	We will show that condition~\ref{cond:locked} with subcondition~\ref{cond:locked-source} will be true at time $t+1$.
	If $\varline{v}{t}{alg:timer-zero-cond}{\timer} \neq 0$ for all nodes, then by induction there is a node such that $\varline{v}{t}{alg:update-energy-cond}{\source = \true}$.
	Otherwise, if $\varline{v}{t}{alg:timer-zero-cond}{\timer} = 0$ but the signal is present, then all nodes run line~\ref{alg:locked-check-signal} and some node will witness the signal and set $\varline{v}{t}{alg:locked-signal-true}{\source} = \true$.
	Thus, we will again have some node with $\varline{v}{t}{alg:update-energy-cond}{\source = \true}$.
	In any case, there will be a node with $v_{t+1}.\source = \true$.
	Note that $|V_t^{\leq m}| \leq n = \max(0, m - m + n)$ trivially.
	Since there is at least one node with $v_{t+1}.\source$, and thus $v_{t+1}.\ttl = v_{t+1}.\maxttl = m$ by Lemma~\ref{lem:source-ttl}, we have $|V_t^{\leq m-1}| \leq n-1 = \max(0, (m-1)-m + n)$.
	Then conditions~\ref{cond:all-maxttl}, \ref{cond:all-locked}, and \ref{cond:timer-sync} are true and the case of $i = 0, \ldots, m-2$ for condition~\ref{cond:locked-source} can be proved exactly as they were for case 3 on the nodes that execute line~\ref{alg:update-energy}.
\end{itemize}

Now suppose that condition~\ref{cond:setting} was true at time $t$.
We once again split into cases.
\begin{itemize}
	\item
	\textit{Case 1}.
	$s^*(t) \leq m - 2$.
	
	We will show that condition~\ref{cond:setting} will continue to be true at time $t+1$.
	Clearly no node runs line~\ref{alg:increase-energy-locked} since $v_t.\state \neq \locked$ for all $v$ by induction.
	Thus, condition~\ref{cond:small-maxttl} will be true at $t+1$.
	Note that $s^*(t+1) = s^*(t) + 1$ since all nodes $v \in S_t$ have $\varline{v}{t}{alg:timer-sync}{\timer} = s^*(t)$ and thus $\varline{v}{t}{alg:timer-increment}{\timer} = s^*(t) + 1 \mod m = s^*(t) + 1$.
	Also, since no node has $\timer = m-1$, $\varline{v}{t}{alg:timer-increment}{\timer} \neq 0$ for all $v$. 
	Thus no nodes will execute line~\ref{alg:become-locked}, and condition~\ref{cond:not-locked} is true at time $t+1$.
	By induction, there is a node $v \in S_t$ with $v_t.\source = \true$.
	By tracing the algorithm, it is clear that $v_{t+1}.\source = \true$ and thus condition~\ref{cond:set-source} is true at $t+1$.
	Let $U = \{u \in V : u \notin S_t ~\land~ \exists v \in S_t [uv \in E_t]\}$.
	Note that $S_{t+1} = S_t \cup U$.
	If $S_t = V$ then $S_{t+1} = V$ and condition~\ref{cond:num-set} is obviously true at $t+1$.
	Otherwise, if $S_t \neq V$, then $U$ must be nonempty since $G$ is connected.
	Then $|S_{t+1}| = |S_t| + |U| \geq \min(n, s^*(t)) + 1 \geq \min(n, s^*(t) + 1) = \min(n, s^*(t+1))$ and so condition~\ref{cond:num-set} is true at $t+1$ in this case too.
	
	Finally, we prove that condition~\ref{cond:num-set-ttls} will be true at time $t+1$.
	Note that the condition is true for $i = m$ because $|S_{t+1}^{\leq m}| \leq |S_{t+1}| = \max(0, m - m + |S_{t+1}|)$.
	Also, since we have shown there is a node $v \in S_{t+1}$ with $v_{t+1}.\source = \true$ and thus $v_{t+1}.\ttl = m$ by Lemma~\ref{lem:source-ttl}, we must have $|S_{t+1}^{\leq m -1}| \leq |S_{t+1}| - 1 = \max(0, (m-1) - m + |S_{t+1}|)$ and thus the condition is true for $i = m-1$.
	Finally, let $i = 0, \ldots m - 2$.
	As in previous cases, note that $S_{t+1}^{\leq i} \subseteq S_t^{\leq i + 1} \cup U$.
	If $S_t^{\leq i + 1} = \emptyset$ then note that no node can have $l^* \leq i + 1$ in round $t$.
	Thus, all nodes must have $v_{t+1}.\ttl \geq l^* - 1 \geq i + 1$ so $|S_{t+1}^{\leq i}| = 0$ so the condition is true for $i$.
	Otherwise, by induction, $1 \leq |S_t^{\leq i+1}| \leq \max(0, i +1- m + |S_t|) \leq |S_t| - 1$.
	Thus, there must be nodes $v \in S_t^{\leq i+1}$ and $w \notin S_t^{\leq i+1}$ such that $vw \in E_t$.
	This would imply $v \notin S_{t+1}^{\leq i}$ and thus $S_{t+1}^{\leq i} \subseteq S_t^{\leq i+1} \setminus \{v\} \cup U$.
	Thus, $|S_{t+1}^{\leq i}| \leq |S_t^{\leq i+1}| - 1 + |U| \leq (i+1-m+|S_t|) -1 + |U| = i-m+|S_{t+1}| \leq \max(0, i-m+|S_{t+1}|)$ so condition~\ref{cond:num-set-ttls} is true.
	
	\item
	\textit{Case 2}.
	$s^*(t) = m - 1$. 
	
	We argue that condition~\ref{cond:locked} will be true at time $t+1$.
	Note that $m - 1 \geq n - 1$.
	Thus, condition~\ref{cond:num-set} says that $|S_t| \geq \min(n, m -1) \geq n - 1$.
	Thus, if we define $U$ as in case 1, either $U$ is empty or a singleton set.
	In any case, if $u \in U$, after line~\ref{alg:increase-energy-recv}, it must have $u_{t+1}.\maxttl = m$ since all other nodes are in $S_t$ and thus also have $v_{t+1}.\maxttl = m$.
	Thus, condition~\ref{cond:all-maxttl} will be true at time $t+1$.
	After line~\ref{alg:timer-sync}, all nodes $u \in U$ have $\varline{u}{t}{alg:timer-sync}{\timer} = s^*(t) = m - 1$.
	Thus timers will remain synced, so condition~\ref{cond:timer-sync} will be true at time $t+1$.
	Note that all nodes will have $\varline{v}{t}{alg:timer-increment}{\timer} = 0$.
	Then all nodes will become and stay locked after line~\ref{alg:become-locked}, so condition~\ref{cond:all-locked} will be true at time $t+1$.
	We now split into two more subcases depending on if there is a signal in round $t$.
	
	First, suppose the signal is present in round $t$.
	Then some node will have $\varline{v}{t}{alg:locked-signal-true}{\source} = \true$.
	Thus $v_{t+1}.\source = \true$ so we must prove condition~\ref{cond:locked-source} is true.
	When $i=m$ the condition is trivial to prove, as in other cases.
	By Lemma~\ref{lem:source-ttl}, $v_{t+1}.\ttl = m$.
	Thus, $|V_t^{\leq m-1}| \leq n-1 \leq \max(0, i-m+n)$ so the condition is true for $i = m-1$.
	For $i = 0\ldots m-2$ we can argue exactly as we did in case 1 to show that the condition holds.
	
	Finally, suppose the signal is not present in round $t$.
	Then all nodes will have $\varline{v}{t}{alg:locked-signal-false}{\source} = \false$ and thus $v_{t+1}.\source = \false$.
	In this case we must prove that condition~\ref{cond:locked-no-source} is true at time $t+1$.
	By Lemmas~\ref{lem:source-ttl} and~\ref{lem:ttl-bound}, all nodes have $v_{t+1}.\ttl \leq m - 1 = m - v_{t+1}.\timer - 1$.
	Thus, condition~\ref{cond:large-ttl} is true at $t+1$.
	Finally condition~\ref{cond:small-ttl} only needs to be proved for $i = 0,\ldots m-2$, which can be done in the same way as case 1.
\end{itemize}
Recall that we had previously defined $t_0$ so that there exists a $v$ such that $\varline{v}{t_0}{alg:increase-energy-recv}{\maxttl} = m$ and $v_{t_0}.\maxttl = 2m \geq 2n$.
In order to run line~\ref{alg:increase-energy-locked}, $v$ must have received from a node $u$ with $u_{t_0}.\maxttl = m$ and $u_{t_0}.\state = \locked$.
Also, either $v_{t_0}.\state = \reset$ or $v_{t_0}.\maxttl < m$.
However, $u_{t_0}.\state = \locked$ implies that condition~\ref{cond:locked} is true.
Thus, by conditions~\ref{cond:all-maxttl} and~\ref{cond:all-locked} $v_{t_0}.\maxttl = m$ and $v_{t_0}.\state = \locked$, a contradiction.
\end{proof}

\end{document}